\documentclass[11pt]{amsart}  

\usepackage{amscd} 
\usepackage{mathrsfs}

\newtheorem{theorem}{Theorem}[section]
\newtheorem{lemma}[theorem]{Lemma}
\newtheorem{corollary}[theorem]{Corollary} 
\theoremstyle{definition}

\newtheorem{remark}[theorem]{Remark} 
\numberwithin{equation}{section}

\def\z*{\bar z}

\def\B{\mathsf B}

\def\uno{\mathsf 1}

\def\F{F}

\def\C{\mathcal C}

\def\dom{\text{\rm dom}}
\def\ran{\text{\rm ran}}
\def\supp{\text{\rm supp}}

\def\G{\mathscr G}
\def\RE{\mathbb R}
\def\CO{{\mathbb C}}

\def\SL{S\! L}
\def\DL{D\! L}

\def\Sf{\mathbb S}

\def\ph*{\phi_\star}

\def\t {\tilde}
\def\be{\begin{equation}}
\def\ee{\end{equation}}
\def\min{{\rm min}}
\def\max{{\rm max}}
\def\-{{\rm in}}
\def\+{{\rm ex}}
\def\comp{{\rm comp}}
\def\loc{{\rm loc}}
\def\n{|\!|\!|}

\begin{document}
\title[Limiting Absorption Principle, Eigenfunctions and Scattering]{Limiting Absorption Principle, Generalized Eigenfunctions and Scattering Matrix for Laplace Operators with Boundary conditions on Hypersurfaces}
\author{Andrea Mantile}
\author{Andrea Posilicano}
\author{Mourad Sini}

\address{Laboratoire de Math\'{e}matiques, Universit\'{e} de Reims -
FR3399 CNRS, Moulin de la Housse BP 1039, 51687 Reims, France}
\address{DiSAT, Sezione di Matematica, Universit\`a dell'Insubria, via Valleggio 11, I-22100
Como, Italy}
\address{RICAM, Austrian Academy of
Sciences, Altenbergerstr. 69, A-4040 Linz, Austria}
\email{andrea.mantile@univ-reims.fr}
\email{andrea.posilicano@uninsubria.it}
\email{mourad.sini@oeaw.ac.at}

\begin{abstract}We provide a limiting absorption principle for the self-adjoint realizations of Laplace operators corresponding to boundary conditions on (relatively open parts $\Sigma$ of) compact hypersurfaces $\Gamma=\partial\Omega$, $\Omega\subset{\mathbb{R}}^{n}$. For any of such self-adjoint operators we also provide the generalized eigenfunctions and the scattering matrix; both these objects  are written in terms of operator-valued Weyl functions.  We make use of a Kre\u\i n-type formula which provides the resolvent difference between the operator corresponding to self-adjoint boundary conditions on the hypersurface and the free Laplacian on the whole space ${\mathbb{R}}^{n}$. Our results apply to all standard examples of boundary conditions, like Dirichlet, Neumann, Robin, $\delta$ and $\delta'$-type, either assigned on $\Gamma$ or on $\Sigma\subset\Gamma$. 
\end{abstract}

\maketitle

\begin{section}{Introduction}
Given an open bounded set $\Omega\subset\RE^{n}$ with smooth boundary
$\Gamma$, let $\Delta^{\!\circ}$ denote the not positive symmetric operator in $L^{2}(\RE^{n})$ given by the restriction of the Laplacian to $C^{\infty}_{\comp}(\RE^{n}\backslash\Gamma)$. In the recent paper \cite{MaPoSi}, we provided the complete family of self-adjoint extensions of $\Delta^{\!\circ}$ and a Kre\u\i n-type formula giving the resolvent difference between any extension and the  self-adjoint (free) Laplacian $\Delta$ with domain $H^{2}(\RE^{n})$ (we recall these results in Theorem \ref{Theorem_Krein}).
Some sub-families of extensions have been considered in \cite{BEK} and   \cite{ER} by a quadratic form approach and in \cite{BLL}
by quasi boundary triple theory.
In particular, in \cite[Section 4]{BLL}, Schatten-von Neumann estimates for the difference of the powers of the resolvent of the free and self-adjoint extensions corresponding to $\delta$-type boundary conditions (supported either on $\Gamma$ or on $\Sigma\subset\Gamma$) and $\delta'$-type ones  (supported on $\Gamma$) are provided; these give existence and completeness of the wave operators of the associated scattering systems. In \cite[Theorems 4.10 and 4.11]{MaPoSi} we extended such kind of  Schatten-von Neumann estimates to a larger class containing, for example, self-adjoint extensions corresponding to Dirichlet, Neumann, Robin, $\delta$ and $\delta'$-type conditions, either assigned on $\Gamma$ or on $\Sigma\subset\Gamma$, where $\Sigma$ is relatively open with a Lipschitz boundary. To this concern we recall that estimates for the difference of the powers of the resolvents  and their applications to scattering in exterior domains first appeared in the pioneering work by Birman \cite{Birman}.\par 
Let us stress that, by the decomposition $\RE^{n}=\Omega_{\-}\cup \Gamma\cup\Omega_{\+}$, $\Omega_{\-}\equiv\Omega$, $\Omega_{\+}:=\RE^{n}\backslash\overline\Omega$, one has $\Delta^{\!\circ}=\Delta^{\min}_{{\-}}\oplus\Delta^{\min}_{{\+}}$ and so one could obtain all self-adjoint extensions of $\Delta^{\!\circ}$ corresponding to separating boundary conditions by using the results (obtained by Grubb in \cite{Grubb}, building on previous work by Birman, Kre\u\i n and Vi\u sik, see \cite{[B]}, \cite{[K3]}, \cite{[V]}) providing the whole family of self-adjoint extensions of $\Delta^{\min}_{{\-/\+}}$; however such construction, broadened to include all self-adjoint extensions, would lead to a Kre\u\i n-type formula giving the resolvent difference between an extension $\widehat \Delta$ and the direct sum of the interior and exterior Dirichlet Laplacians $\Delta^{D}_{{\-}}\oplus\Delta^{D}_{{\+}}$. This is not the right operator since we are interested in the study of the scattering system $(\widehat \Delta,\Delta)$, where $\Delta$ denotes the free Laplacian on the whole $\RE^{n}$; whenever one considers  semi-transparent boundary conditions (as the ones considered in \cite{BLL} and in Section 6, Examples \ref{delta} and \ref{deltaprimo}), or boundary conditions assigned only on $\Sigma\subset\Gamma$ (see Section 7), the choice of the Laplacian $\Delta$ as the operator representing the free propagation is the most natural one.  
\par
The first aim of this paper is to show that the Limiting Absorbtion Principle (LAP for short)  holds for an ample class of self-adjoint extensions of the symmetric operator $\Delta^{\!\circ}$. This is accomplished by applying  abstract results due to Walter Renger 
(see \cite{Ren1} and \cite{Ren2}) to our Kre\u\i n-type resolvent formula (see Theorem \ref{Theorem_LAP}). As usual, LAP implies the absence of singular continuous spectrum (see Corollary \ref{sing}). Even if interesting by itself, the result about the validity of LAP  does not show that the resolvent Kre\u\i n formula itself survives in the limit. Such a  limit Kre\u\i n's resolvent formula is here provided in next Theorem \ref{teo_LAP}. With such results at hands, in Section 5  we construct, for a given self-adjoint extension $\widehat\Delta$ of $\Delta^{\circ}$, the couple of families of generalized eigenfunctions $u^{\pm}_{\xi}$ related to the plane waves $u^{\circ}_{\xi}(x)=e^{i\,\xi\cdot x}$ with incoming ($+$) or outgoing ($-$) radiation conditions. Such eigenfunctions then allow to define the corresponding Fourier type transforms $F_{\pm}$ which diagonalize the self-adjoint extension; the wave operators for the scattering system $(\widehat\Delta,\Delta)$ are then given by $W_{\pm}=F^{*}_{\pm}F$, where $F$ denotes the ordinary Fourier transform (see Theorem 5.4). Both the eigenfunctions $u^{\pm}_{\xi}$ and the Fourier transforms $F_{\pm}$ are expressed in terms of the operator-valued Weyl functions appearing in the limit Kre\u\i n resolvent formula given in Theorem \ref{teo_LAP}.  Finally, in Theorem \ref{SM}, using again the operator-valued Weyl functions, we provide the  kernel (proportional to the scattering amplitude) of $\uno-S_{k}$, where $S_{k}$ is the on-shell scattering operator. 
\par 
In Sections 6 and 7, we show that our LAP-based results can be applied to a wide class of self-adjoint operators which includes self-adjoint realizations of the Laplacian with Dirichlet, Neumann, Robin, $\delta$ and $\delta'$-type boundary conditions assigned either on the whole $\Gamma$ or on a relatively open subset with Lipschitz boundary $\Sigma\subset\Gamma$. We provide a representation of the scattering matrix $S_{k}$ in terms of operator-valued Weyl functions evaluated on the traces at ($\Sigma\subset$)$\Gamma$ of the plane waves $u^{\circ}_{\xi}$.
\par
Our time-independent approach  has been inspired by the work by  Albeverio, Brasche and Koshmanenko \cite{ABK}, where LAP and Lippman-Schwinger 
equations are studied for finite-rank singular perturbations, and can be interpreted as an extension to the case of general boundary conditions and hypersurfaces of the paper \cite{IS} by Ikebe and Shimada concerning $\delta$-type boundary conditions on a sphere (see also \cite{S}). An alternative abstract approach, which do not use LAP but directly exploits  the existence of limiting operator-valued Weyl-functions, has been developed in \cite{BMN1}, \cite{BMN2} and \cite{BMN3} by Behrndt, Malamud and Neidhardt (the first two works concern the finite-rank case; see also \cite{AP},\cite[Chapter 4]{AK}). In particular,  in the recent paper \cite{BMN3}, a representation of the scattering matrix in term of operator-valued Weyl functions is provided for couples of self-adjoint extension of a given symmetric operator under the hypothesis that their resolvent difference is trace-class. 
In our less abstract setting, which applies to Laplacian with boundary conditions on $(\Sigma\subset)\Gamma$, we do not need the trace-class condition and the  results hold in any  dimension. 
\par
Let us remark that, once LAP and a Kre\u\i n's limit formula have been attained (see Theorems \ref{LAP} and \ref{teo_LAP}), a representation formula for the scattering matrix $S_{k}$ can be obtained by using the Birman-Yafaev general scheme in stationary scattering theory (see  e.g. \cite{BY}, \cite{Y}, \cite{Y1}) together with the Birman-Kato invariance principle applied to the resolvent operators (see Remark \ref{BY} for more details). However we preferred to present here a less abstract proof following the classical scheme used in potential scattering theory (see e.g. \cite{Agm}, \cite{AS}, \cite{Ike}, \cite{Sch}  and references therein).      
\par
We conclude the introduction with some remarks about our smoothness hypotheses on $\Gamma$. Such an hypothesis  gives the existence of the wave operators (see \cite[Theorems 4.11 and 4.12]{MaPoSi}) through asymptotic estimates on the eigenvalues of the Laplace-Beltrani operator on $\Gamma$ (see \cite[Lemma 4.7]{BLL1}). These estimates, obtained using pseudodifferential operator techniques,  require smoothness; we presume that asymptotic estimates of this kind  hold under a weaker $C^{1,1}$ (or at least $C^{2}$) hypothesis, but we did not find any proof of that in the literature. Since our result concerning existence of LAP does not require any smoothness hypothesis, conditional on the existence of wave operators, the general results here presented hold in the case  $\Gamma$ is  an hypersurface of class $C^{1,1}$, as for the results presented in \cite{MaPoSi}, while, as regards the explicit examples given in Section 7 considering boundary conditions on $\Sigma\subset \Gamma$, one needs more regularity (of the kind $C^{k,1}$, where $k>1$ depends on the kind of boundary conditions used, see \cite[Section 6]{MaPoSi}). In the series of papers \cite{EP1}-\cite{EP3}, limited to the case in which $n=2$ and Dirichlet or Neumann boundary conditions are assigned on the whole $\Gamma$, the authors provided a resolvent formula and a representation for the scattering matrix of the same kind of the ones here given in Examples 6.1 and 6.2, only assuming that the boundary $\Gamma$ is a piecewise smooth curve. This suggests that our results, which hold for a quite larger class of boundary conditions, could be extended to include the case in which $\Gamma$ is a planar curvilinear polygon (see e.g. \cite{D} and \cite{G} for elliptic boundary value problem in not smooth domains).
\end{section}

{\bf Acknowledgments.} We thank the anonymous referees for the stimulating remarks, for the useful bibliographic suggestions and in particular for inspiring Remark \ref{BY}. \par
The authors were partially supported by the Austrian Science Fund 
(FWF):[P28971-N32]

\begin{section}{Preliminaries\label{Sec_Surfop}}

\subsection{Trace maps and boundary-layer operators } Here we recall some definitions and results about Sobolev spaces on subset of $\RE^{n}$ and single and double layer operators on their boundaries (see e.g. \cite{LiMa1} and \cite{McLe}).\par
Given $\Omega\subset\mathbb{R}^{n}$ open and bounded, with smooth boundary $\Gamma$, we adopt the notation: $\Omega_{\-}=\Omega$, $\Omega_{\+}=\mathbb{R}%
^{n}\backslash\overline\Omega$, while $\nu$ is the exterior unit normal to
$\Gamma$.  $H^{s}(\RE^{n})$, $H^{s}(\Omega_{\-})$, $H^{s}(\Omega_{\+})$, $H^{s}(\Gamma)$, $s\in\RE$, denote the usual scales of Sobolev-Hilbert spaces of function on $\RE^{n}$, $\Omega_{\-}$, $\Omega_{\+}$ and $\Gamma$ respectively. The zero
and first-order traces on $\Gamma$ are defined on smooth functions as%
\begin{equation}%
\gamma_{0}u=\left.  u\right\vert _{\Gamma}\,, \qquad \gamma_{1}u=
\nu\cdot\nabla u|
{\Gamma}\,,
\label{trace}%
\end{equation}
and extend to the bounded linear operators 
\begin{equation}%
\gamma_{0}\in{\B}(   H^{2}(  \mathbb{R}^{n})
,H^{{\frac32}}(   \Gamma)  )  \,, \qquad \gamma_{1}\in{\B}%
(   H^{2}(  \mathbb{R}^{n})  ,H^{{\frac12}}(   \Gamma)
)  \,.
\label{trace_est}%
\end{equation}
We use the symbol $\Delta$ to denote the distributional Laplacian; its restriction to $H^{2}(\RE^{n})$
$$\Delta: H^{2}(\RE^{n})\subset L^{2}(\RE^{n})\to L^{2}(\RE^{n})$$ 
gives rise to a self-adjoint operator which describes the free propagation of waves in the whole space $\RE^{n}$; this will be our reference operator. \par 
For $z\in\CO\backslash(-\infty,0]$, the single and double layer operators are defined by %
\begin{equation}%
\SL_{z}=( \gamma_{0}(   -\Delta+\bar z)  ^{-1})^{\ast}\,, \qquad \DL_{z}=( \gamma_{1}( 
-\Delta+\bar z)  ^{-1})^{\ast}\,,
\label{Layer_op}%
\end{equation}
and by duality there follows%
\begin{equation}%
\SL_{z}\in{\B}(  H^{-{\frac32}}(  \Gamma)  ,L^{2}(
\mathbb{R}^{n}) )  \,, \quad  \DL_{z}\in{\B}(
H^{-{\frac12}}(  \Gamma)  ,L^{2}(  \mathbb{R}^{n}))
\,.
\label{Layer_op_est}%
\end{equation}
The integral kernel $R_{z}(x,y)$ of the resolvent 
$(-\Delta+z)^{-1}$,  $z\in\CO\backslash(-\infty,0]$, is given by $R_{z}(x,y)=\G_{z}(x-y)$, where  
\begin{align}\label{green}
\G_{z}(x)
=\frac{1}{2\pi}\,\left(\frac{\sqrt z}{2\pi\|x\|}\right)^{\!\frac{n}2-1}\!\!\!\!K_{\frac{n}2-1}(\sqrt z\,\|x\|)
\,,\qquad \text{Re}(\sqrt z)>0\,,
\end{align}
and 
$K_{\alpha}$ denotes 
the modified Bessel functions of second kind of order $\alpha$. Thus, for $x\notin\Gamma$ and $\phi,\varphi\in L^{2}(\Gamma)$, one has 
\begin{equation*}
\SL_{z}\phi(x)=\int_{\Gamma}\G_{z}(x-y)\,\phi(y)\,d\sigma (y)\,, \label{S_Gamma}%
\end{equation*}
and%
\begin{equation*}
\DL_{z} \varphi(x)=\int_{\Gamma}\nu(y)\!\cdot\!\nabla \G_{z}(x-y)\,\varphi(y)\,d\sigma (y)\,,\label{D_Gamma}%
\end{equation*}
where $\sigma $ denotes the surface measure. \par Let us define $\gamma\in\B(H^{2}(\RE^{n}),H^{{\frac32}}(  \Gamma)\oplus H^{{\frac12}}(  \Gamma) $ by
$$
\gamma u:=\gamma_{0}u\oplus\gamma_{1}u
$$
and, for any $z\in\CO\backslash(-\infty,0]$, 
$G_{z}\in \B(H^{-{\frac32}}(  \Gamma)\oplus H^{-{\frac12}}(  \Gamma) ,L^{2}(
\mathbb{R}^{n}) )$ by 
$$
G_{z}:=( \gamma(   -\Delta+\bar z)  ^{-1})^{\ast}\,;
$$
equivalently
\begin{equation}\label{Gz}
G_{z}(\phi\oplus\varphi):=\SL_{z}\phi+\DL_{z}\varphi\,.
\end{equation}
For any $z\in\CO\backslash(-\infty,0]$ and for any $\phi\oplus\varphi\in H^{-{\frac32}}(  \Gamma)\oplus H^{-{\frac12}}(  \Gamma)$ one has
$$
G_{z}(\phi\oplus\varphi)\in{\C^{\infty}}(\RE^{n}\backslash\Gamma)\quad\text{and}\quad((\Delta-z)G_{z}(\phi\oplus\varphi))(x)=0\,,\quad x\in \RE^{n}\backslash\Gamma\,.
$$
The one-sided trace maps $$\gamma_{i}^{\natural}\in{\B}(   H^{2}( 
\Omega_{\natural})  ,H^{{\frac32}-i}(   \Gamma)  )\,,\quad\natural=\-,\+\,,\quad i=0,1\,,
$$ defined on smooth (up to the boundary) functions by%
\begin{equation}%
\gamma_{0}^{\natural}u_{\natural}=u_{\natural}|\Gamma\,, \quad\gamma^{\natural}_{1}u_{\natural}=\nu\cdot\nabla u_{\natural}|\Gamma\,,\quad \natural=\-,\,\+\,,
\end{equation}
can be extended to $$\hat\gamma_i^{\natural}\in{\B}(  H^{0}_{\Delta} (\Omega_{\natural})),H^{-{\frac12}-i}(   \Gamma)  )  \,,\quad\natural=\-,\+\,,\quad i=0,1\,,$$ 
where 
$$
H^{0}_{\Delta} (\Omega_{\natural}):=\{u_{\natural}\in L^{2}(\Omega_{\natural}): \Delta u_{\natural}\in L^{2}(\Omega_{\natural})\}\,,
$$
$$
\|u_{\natural}\|^{2}_{H^{0}_{\Delta} (\Omega_{\natural})}:=\|\Delta u_{\natural}\|^{2}_{L^{2}(\Omega_{\natural})}
+\|u_{\natural}\|^{2}_{L^{2}(\Omega_{\natural})}\,.
$$
Setting $\Delta^{\max}_{\natural}:=\Delta|H^{0}_{\Delta}(\Omega_{\sharp})$, by the   ''half'' Green formula (see \cite[Theorem 4.4]{McLe}), one has, for any $u\in H^{1}(\Omega_{\natural})\cap H^{0}_{\Delta}(\Omega_{\natural})$,
\begin{align}\label{hG}
\langle -\Delta^{\max}_{\natural}u_{\natural},u_{\natural}\rangle_{L^{2}(\Omega_{\natural})}
=\|\nabla u_{\natural}\|^{2}_{L^{2}(\Omega_{\natural})}
+\epsilon_{\natural}\langle\hat\gamma^{\natural}_{1}u_{\natural},\gamma^{\natural}_{0}u_{\natural}\rangle_{-\frac12,\frac12}\,,\quad
\end{align}
$$
\epsilon_{\natural}=\begin{cases}-1\,,&\natural=\-\\
+1\,,&\natural=\+\,.\end{cases}
$$
Setting 
$$
H^{0}_{\Delta} (\RE^{n}\backslash\Gamma):=H^{0}_{\Delta} (\Omega_{\-})\oplus H^{0}_{\Delta} (\Omega_{\+})\,,
$$
the extended maps $\hat\gamma_{i}^{\natural}$ allow to define the bounded maps 
$$
\hat\gamma_{i}\in{\B}(H^{0}_{\Delta} (\RE^{n}\backslash\Gamma),H^{-{\frac12}-i}(   \Gamma)  )\,,\quad i=0,1\,, $$
$$
[\hat\gamma_{i}]\in{\B}(H^{0}_{\Delta} (\RE^{n}\backslash\Gamma),H^{-{\frac12}-i}(   \Gamma)  )
\,,\quad i=0,1\,, $$
by
\begin{equation}
\hat\gamma_{i}u:=\frac{1}{2}\big(   \hat\gamma_{i}^{\-}(u|\Omega_{\-})+\hat\gamma_{i}^{\+}(u|\Omega_{\+})\big)
\,,
\label{trace_def}%
\end{equation}
and \begin{equation}
[  \hat\gamma_{i}] u =\hat\gamma_{i}^{\+}(u|\Omega_{\+})-\hat\gamma_{i}^{\-}(u|\Omega_{\-})\,,\quad
\,. \label{jumps}%
\end{equation}
Notice that the maps $\gamma_{i}\in \B(H^{2}(\RE^{n}\backslash\Gamma), H^{\frac32-i}(\Gamma))$ defined by 
$$
\gamma_{i}:=\hat\gamma_{i}|H^{2}(\RE^{n}\backslash\Gamma)\,,\quad i=0,1\,,\quad H^{2}(\RE^{n}\backslash\Gamma):=H^{2}(\Omega_{\-})\oplus H^{2}(\Omega_{\+})\,,
$$ 
coincide with the ones in   \eqref{trace} when restricted to $H^{2}(\RE^{n})$.\par
The corresponding extension of the trace map $\gamma$ is 
\begin{equation}\label{tau}
\hat\gamma  \in{\B}(  H^{0}_{\Delta} (\RE^{n}\backslash\Gamma),H^{-{\frac12}}(   \Gamma)\oplus H^{-{\frac32}}(   \Gamma)  )\,,
\quad
\hat\gamma   u:=(\hat\gamma_{0}u)\oplus(\hat\gamma_{1}u)\,,
\ee
while the related jump map is 
\be\label{[tau]}
[\hat\gamma  ]\in{\B}( H^{0}_{\Delta} (\RE^{n}\backslash\Gamma),H^{-{\frac32}}(   \Gamma)\oplus H^{-{\frac12}}(   \Gamma)  )\,,
\qquad
[\hat\gamma  ]u
=(-[  \hat\gamma_{1}]u)  \oplus([  \hat\gamma_{0}]u)  \,.
\end{equation}
Using the definition \eqref{Gz} and the mapping properties of the layer operators (see \cite[Section 3.4]{MaPoSi} and the references therein), it results%
\begin{equation}
G_{z}\in{\B}(   H^{-{\frac32}}(   \Gamma)\oplus H^{-{\frac12}}(   \Gamma),H^{0}_{\Delta} (\RE^{n}\backslash\Gamma))
 \label{G_z_reg}%
\end{equation}
and so $[\hat \gamma  ]G_{z}\in {\B}(   H^{-{\frac32}}(   \Gamma)\oplus H^{-{\frac12}}(   \Gamma))$.
More precisely, by the well known jumps relations
for the layer operators, 
one has 
\be
[\hat \gamma  ]G_{z}=1_{H^{-{\frac32}}(\Gamma)\oplus H^{-{\frac12}}(   \Gamma)  }\,.
\label{G_z_jumps}%
\end{equation}
\subsection{Weighted spaces  }
We now introduce the family of weighted spaces $L^{2}_{\sigma}(\RE^{n})$ and $H^{2}_{\sigma}(\RE^{n})$, defined, for any $\sigma\in\RE$, by
$$
L^{2}_{\sigma}(\RE^{n}):=\{u\in L^{2}_{\loc} (\RE^{n}):\|u\|_{L^{2}_{\sigma}(\RE^{n})}<+\infty\}\,,
$$
$$
H^{2}_{\sigma}(\RE^{n}):=\{u\in H^{2}_{\loc} (\RE^{n}):\|u\|_{H^{2}_{\sigma}(\RE^{n})}<+\infty\}\,,
$$
$$
\left\Vert u\right\Vert^{2} _{L^{2}_{\sigma}(   \mathbb{R}^{n})
}:=\int_{\RE^{n}}\left(1+\|x\|^{2}\right)^{\sigma}|u(x)|^{2}dx
$$
$$
\left\Vert u\right\Vert^{2} _{H^{2}_{\sigma}(   \mathbb{R}^{n})
}:=\|u\|^{2}_{L^{2}_{\sigma}(   \mathbb{R}^{n})}+\sum_{1\le i\le n}\|\partial_{x_{i}}u\|^{2}_{L^{2}_{\sigma}(   \mathbb{R}^{n})}
+\sum_{1\le i,j\le n}\| \partial^{2}_{x_{i} x_{j}}u\|^{2}_{L^{2}_{\sigma}(   \mathbb{R}^{n})}
\,.
$$
The spaces $L^{2}_{\sigma}(\Omega_{\-})$, $L^{2}_{\sigma}(\Omega_{\+})$ and $H^{2}_{\sigma}(\Omega_{\-})$, $H^{2}_{\sigma}(\Omega_{\+})$ are defined in a similar way. Since $\Omega$ is bounded, one has $L^{2}_{\sigma}(\Omega_{\-})=L^{2}(\Omega_{\-})$,  $H^{2}_{\sigma}(\Omega_{\-})=H^{2}(\Omega_{\-})$, the equalities holding in the Banach space sense; thus 
$$L^{2}_{\sigma}(\RE^{n})=L^{2}(\Omega_{\-})\oplus L^{2}_{\sigma}(\Omega_{\+})$$ and 
$$H^{2}_{\sigma}(\RE^{n}\backslash\Gamma):=H^{2}_{\sigma}(\Omega_{\-})\oplus H^{2}_{\sigma}(\Omega_{\+})=H^{2}(\Omega_{\-})\oplus H^{2}_{\sigma}(\Omega_{\+})\,.
$$
Then the trace operators can be extended to  $H^{2}_{\sigma}(\RE^{n}\backslash\Gamma)$, $\sigma<0$, by  
$$ 
\gamma^{\+}_{0} u_{\+}:=\gamma^{\+}_{0}(\chi u_{\+})\,,\quad\gamma^{\+}_{1} u_{\+}:=\gamma^{\+}_{1}(\chi u_{\+})\,,$$
where $\chi$ belongs to $\C^{\infty}_{\comp}(\Omega^{c})$ and $\chi=1$ on a neighborhood of $\Gamma$. 
\begin{remark} In the following, we use the shorthand notation $\langle\cdot,\cdot\rangle$ to denote both the dualities   $\big(H^{-s_{1}}(\Gamma)\oplus H^{-s_{2}}(\Gamma)\big)$-$\big(H^{s_{1}}(\Gamma)\oplus H^{s_{2}}(\Gamma)\big)$ and $L^{2}_{-\sigma}(\RE^{n})$-$L^{2}_{\sigma}(\RE^{n})$.  
\end{remark}

\end{section}
\begin{section}{Self-adjoint Laplace operators with boundary conditions on hypersurfaces}

Let us consider the restriction $\Delta|\ker(\gamma  )$. Since
the kernel of $\gamma  $ is dense in $L^{2}(   \mathbb{R}^{n})  $,
$\Delta|\ker(\gamma  )$ is densely defined, closed and symmetric. Following the
construction developed in \cite{MaPoSi} (to which we refer for more details and proofs), we next provide all the self-adjoint extensions of
$\Delta|\ker(\gamma  )$. The adjoint operator $(\Delta|\ker(\gamma  ))^{\ast}$
identifies with%
\begin{equation*}
\dom(   (\Delta|\ker(\gamma  ))^{\ast})  =H^{0}_{\Delta} (\RE^{n}\backslash\Gamma)\,,\quad 
(\Delta|\ker(\gamma  ))^{\ast}=\Delta^{\max}_{\-} \oplus \Delta^{\max}_{\+}\,.
\end{equation*} 
An alternative representation of $(\Delta|\ker(\gamma  ))^{*}$ is given by (see
\cite[Lemma 2.3 and Lemma 4.2]{MaPoSi})%
\begin{align*}
 &\dom((\Delta|\ker(\gamma  ))^{\ast})\\  =&\left\{  u=u_{\circ}+G(\phi\oplus\varphi):
u_{\circ}\in H^{2}(   \mathbb{R}^{n})\,,\ \phi\oplus\varphi\in
H^{-{\frac32}}(   \Gamma)\oplus H^{-{\frac12}}(   \Gamma)  \right\} \\
\equiv&\left\{  u\in L^{2}(\RE^{n}):u_{\circ}:=\big(u+\SL[\hat\gamma_{1}]u-\DL[\hat\gamma_{0}]u\big)\in H^{2}(   \mathbb{R}^{n})  \right\} 
\,,
\end{align*}
\be\label{adjoint3}
(\Delta|\ker(\gamma  ))^{\ast}u=\Delta u_{\circ}+G(\phi\oplus\varphi)=
\Delta u-[\hat \gamma_{1}]u\,\delta_\Gamma-[\hat \gamma_{0}]u\,\nu\!\cdot\!\nabla\delta_\Gamma\,,
\ee
where $G:=G_{1}$
and the Schwartz distribution $\delta_{\Gamma}$ is defined by $\delta_{\Gamma}(u):=\int_{\Gamma}u(x)\,d\sigma (x)$. 
\par 
Given an orthogonal projection $\Pi:H^{{\frac32}}(   \Gamma)\oplus H^{{\frac12}}(   \Gamma)\to H^{{\frac32}}(   \Gamma)\oplus H^{{\frac12}}(   \Gamma)$, the
dual map $\Pi^{\prime}:H^{-{\frac32}}(   \Gamma)\oplus H^{-{\frac12}}(   \Gamma)\to H^{-{\frac32}}(   \Gamma)\oplus H^{-{\frac12}}(   \Gamma)$
is an orthogonal projection  as well and $\ran(\Pi')=\ran(\Pi)'$.
We say that the densely defined linear operator  $$\Theta:\dom(   \Theta)  \subseteq
\ran(   \Pi)^{\prime}  \rightarrow \ran(   \Pi)  $$ is
self-adjoint whenever $\Theta=\Theta'$. Let the unitary maps $\Lambda^{s}$ represent the duality mappings on $H^{s/2}(\Gamma)$ onto $H^{-s/2}(\Gamma)$; an explicit representation of $\Lambda^{s}$ is given by $\Lambda^{s}=(-\Delta_{LB}+1)^{s/2}$, where $\Delta_{LB}$ denotes the Laplace-Beltrami operator on $\Gamma$. Equivalently $\Theta$ is self-adjoint whenever the operator $\tilde{\Theta}=\Theta( 
\Lambda^{3}\oplus\Lambda)  $, $\dom(   \tilde{\Theta})
=(   \Lambda^{3}\oplus\Lambda)^{-1}  \dom(   \Theta)
$, is a self-adjoint operator in the Hilbert space $\ran(\Pi)$. 
\par
We define  the operator-valued Weyl function 
 $$\CO\backslash(-\infty,0]\ni z\mapsto M_{z}\in\B(H^{{-\frac32}}(   \Gamma)\oplus H^{{-\frac12}}(   \Gamma);H^{{\frac32}}(   \Gamma)\oplus H^{{\frac12}}(   \Gamma))$$ by $M_{z}:=\gamma  (   G-G_{z})$, i.e., using  the block operator notation,
\be\label{M_z_def}
M_{z}=\left[\,\begin{matrix}\gamma_{0}(\SL-\SL_{z})&\gamma_{0}(\DL-\DL_{z})\\
\gamma_{1}(\SL-\SL_{z})&\gamma_{1}(\DL-\DL_{z})
\end{matrix}\,\right]\,.
\ee
Given the couple $(   \Pi,\Theta)  $, $\Pi$ an orthogonal projector and $\Theta$ self-adjoint,     
define the set%
\begin{equation}
Z_{  \Pi,\Theta  }:=\{z\in\CO\backslash(-\infty,0]:\Theta+\Pi M_{z}\Pi^{\prime
}\in\B(\ran(   \Pi),\ran( 
\Pi)^{\prime} ) \}. \label{Z_Pi_Theta_def}%
\end{equation}
All self-adjoint extensions of $\Delta|\ker(\gamma  )$ are provided by the following theorem (see \cite[Theorem 4.4]{MaPoSi}):
\begin{theorem}
\label{Theorem_Krein} Any
self-adjoint extension of $\Delta|\ker(\gamma  )$ is of the
kind $$\Delta_{\Pi,\Theta }=(\Delta|\ker(\gamma)  )^{\ast}|\dom(\Delta_{\Pi,\Theta  })\,,$$ where $\Pi:H^{{\frac32}}(   \Gamma)\oplus H^{{\frac12}}(   \Gamma)\to H^{{\frac32}}(   \Gamma)\oplus H^{{\frac12}}(   \Gamma)$ is an orthogonal projection, $\Theta
:\dom(   \Theta)  \subseteq \ran(   \Pi)^{\prime}
\rightarrow \ran(   \Pi)  $ is a self-adjoint operator and
\begin{align*}
&\dom(\Delta_ {   \Pi,\Theta  })\\:=&\{u=u_{\circ}+G(\phi\oplus\varphi):u_{\circ}\in
H^{2}(\mathbb{R}^{n}),\, \phi\oplus\varphi\in \dom(   \Theta),\, \Pi\gamma  
u_{\circ}=\Theta(\phi\oplus\varphi)\} \\
=&\{u\in H^{0}_{\Delta}(\RE^{n}\backslash\Gamma):[\hat\gamma]u\in \dom(   \Theta),\, \Pi\gamma  (u+\SL[\hat\gamma_{1}]u-\DL[\hat\gamma_{0}]u\big)
=\Theta[\hat\gamma]u\}
\end{align*}
Moreover $Z_{  \Pi,\Theta  }$ is not void, $\CO\backslash
\mathbb{R}\subseteq Z_{ \Pi,\Theta  }\subseteq\rho(\Delta_ { 
\Pi,\Theta  })$, and the resolvent of the self-adjoint extension
$\Delta_ {  \Pi,\Theta  }$ is given by the Kr\u{e}\i n's type formula%
\begin{equation}\begin{split}
&(   -\Delta_ {   \Pi,\Theta  }+z)  ^{-1}\\
=&(   -\Delta+z)
^{-1}+G_{z}\Pi^{\prime}(   \Theta+\Pi M_{z}\Pi^{\prime})  ^{-1}%
\Pi\gamma  (   -\Delta+z)  ^{-1},\,\quad z\in Z_{   \Pi,\Theta
}\,. \label{Krein}%
\end{split}
\end{equation}

\end{theorem}
\begin{remark} Let us notice that the $\Pi'\,$'s appearing in formula \eqref{Krein} act there merely as the inclusion map $\Pi':\ran(\Pi)^{\prime}\to H^{-3/2}(\Gamma)\oplus H^{-1/2}(\Gamma)$. This means that one does not need to know $\Pi'$ explicitly: it suffices to know the subspace $\ran(\Pi')=\ran(\Pi)'$.
\end{remark}
Given the self-adjoint operator $\Theta:\dom(\Theta)\subseteq \ran(\Pi)'\to \ran(\Pi)$, we now introduce the following assumptions:
\begin{equation}
\dom(   \Theta)  \subseteq 
H^{s_{1}}(\Gamma) \oplus H^{s_{2}}(\Gamma)  \,,\qquad  s_{1}>-\frac32\,,\quad s_{2}>-\frac12\,, \label{Theta_reg1}%
\end{equation}
\begin{equation}
\dom(   \Theta)  \subseteq 
H^{{\frac12}}(\Gamma) \oplus H^{{\frac32}}(\Gamma)   \label{Theta_reg}%
\end{equation}
and
\be\label{Theta_reg2}
\dom(f_{\t\Theta})\subseteq H^{\frac52}(\Gamma)\oplus H^{{\frac32}}(\Gamma)\,, 
\ee
where $f_{\t\Theta}$ is sesquilinear form associated to the self-adjoint operator in $\ran(\Pi)$ defined by $\tilde\Theta:=\Theta(\Lambda^{3}\oplus\Lambda)$. \par
The next result gives informations on the spectrum and scattering of $\Delta_{\Pi,\Theta}$; for the proof of such results we refer to \cite[Lemma 4.10, Corollary 4.12 and Remark 4.14]{MaPoSi}. 
\begin{theorem}\label{TeoSpec} Suppose Theorem \ref{Theorem_Krein} holds. Then:\par\noindent 
1) assumption \eqref{Theta_reg1} implies
$$
\sigma_{ess}(\Delta_{\Pi,\Theta})=(-\infty,0]\,;
$$
2) either assumption \eqref{Theta_reg} or \eqref{Theta_reg2} gives 
the existence and completeness of the wave operators 
$$
W_{\pm}
:=\text{s-}\lim_{t\to\pm\infty}e^{-it\Delta_{\Pi,\Theta}}e^{it\Delta}
\,,\qquad
\t W_{\pm}
:=\text{s-}\lim_{t\to\pm\infty}e^{-it\Delta}e^{it\Delta_{\Pi,\Theta}}P_{ac}\,,
$$ 
i.e. the limits exists everywhere w.r.t. strong convergence, $\ran(W_{\pm})=L^{2}(\RE^{n})_{ac}$, $\ran(\t W_{\pm})=L^{2}(\RE^{n})$ and $W_{\pm}^{*}=\t W_{\pm}$, where $L^{2}(\RE^{n})_{ac}$ denotes the absolutely continuous subspace of $L^{2}(\RE^{n})$ with respect to  $\Delta_ {\Pi,\Theta}$ and $P_{ac}$ is the corresponding orthogonal projector. This then implies
\begin{equation}
\sigma_{ac}(   \Delta_ {   \Pi,\Theta  })  
=(-\infty, 0] \,.
\label{spectrum}%
\end{equation}
\end{theorem}
\begin{remark}
Let us remark that hypothesis \eqref{Theta_reg} holds in the case of global boundary conditions, i.e. assigned on whole boundary $\Gamma$ (see Section 6),  whereas hypothesis \eqref{Theta_reg2} holds in the case of local ones, i.e. assigned on $\Sigma\subset\Gamma$ (see Section 7).  
\end{remark}
\begin{remark}
Let us notice that the apparent discrepancy between the indices in \eqref{Theta_reg} and \eqref{Theta_reg2} is due to the fact that the first one applies to operators acting between the dual pair $(\ran(\Pi)',\ran(\Pi))$ whereas the second one concerns sesquilinear forms in the space $\ran(\Pi)$; when written in terms of $\t\Theta$, condition \eqref{Theta_reg} reads as $\dom(\t\Theta)\subseteq H^{\frac72}(\Gamma)\oplus H^{\frac52}(\Gamma)$. 
\end{remark}
Under hypothesis \eqref{Theta_reg}, it is possible to introduce an alternative description of $\Delta_ {  \Pi
,\Theta  }$  (see \cite[Corollary 4.8]{MaPoSi}):
\begin{corollary}
\label{Lemma_parameter}Let $\Delta_ { 
\Pi,\Theta  }$ be defined according to Theorem
\ref{Theorem_Krein} with $\Theta$ fulfilling (\ref{Theta_reg}). Define%
\begin{equation}
B_{\Theta}:=\Theta+\Pi \gamma G\Pi^{\prime}:\dom(   \Theta)  \subseteq
\ran(   \Pi)^{\prime}  \rightarrow \ran(   \Pi)  \,.
\end{equation}
Then%
\begin{equation}
\dom(   \Delta_ {  \Pi,\Theta  })  =\{u\in H^{2}( 
\mathbb{R}^{n}\backslash\Gamma) : [  \gamma  ]
u\in \dom(   \Theta),\, \Pi\gamma   u=B_{\Theta}[
\gamma  ]  u\}\,,
\end{equation}
and, whenever $z\in Z_{   \Pi
,\Theta  }$,
\begin{equation}
\begin{split}
&(   -\Delta_ {   \Pi,\Theta  }+z)  ^{-1}\\=&(   -\Delta+z)
^{-1}+G_{z}\Pi^{\prime}(   B_\Theta-\Pi\gamma G_{z}\Pi^{\prime})
^{-1}\Pi\gamma(   -\Delta+z)  ^{-1}\,.
\end{split}
\end{equation}
\end{corollary}
The results contained in Theorem \ref{TeoSpec} do not exclude the presence of negative eigenvalues embedded in the essential spectrum, an information that is relevant for the issues to be treated in the next sections. However, since the singular perturbations defining $\Delta_{\Pi,\Theta}$ are compactly supported, an easy application of the unique continuation principle and Rellich's estimate give criteria for the absence of such eigenvalues. For successive notational convenience let us pose 
$$
E^{-}_{\Pi,\Theta}:=\{\lambda\in (-\infty,0):
\lambda\notin \sigma_{p}(\Delta_{\Pi,\Theta})\}\,,
$$
so that absence of negative eigenvalues is equivalent to $E^{-}_{\Pi,\Theta}=(-\infty,0)$.
\begin{theorem}\label{negative1} Let $\Gamma_{\! 0}\subseteq\Gamma$ be a closed set such that 
$\supp(\phi)\cup\supp(\varphi)\subseteq\Gamma_{\! 0}$ for any $\phi\oplus\varphi\in\dom(\Theta)\subseteq\ran(\Pi)^{\prime}$.  
If the open set $\RE^{n}\backslash\Gamma_{\! 0}$ is connected, then $E^{-}_{\Pi,\Theta}=(-\infty,0)$.
\end{theorem}   
\begin{proof} Suppose that there exist $\lambda\in (-\infty,0)$ and $u_{\lambda}\in \dom(\Delta_{\Pi,\Theta})\subseteq L^{2}(\RE^{n})\cap H^{2}_{loc}(\RE^{n}\backslash\Gamma_{\! 0})$ such that $\Delta_{\Pi,\Theta}u_{\lambda}=\lambda\, u_{\lambda}$. Then, by \eqref{adjoint3}, $\Delta u_{\lambda}(x)=\lambda\,u_{\lambda}(x)$ for a.e. $x\in \RE^{n}\backslash\Gamma_{\! 0}$. Thus, by the unique continuation principle (see e.g. \cite[Theorem XIII.63]{R&S}), $u_{\lambda}=0$ a.e. whenever $u_{\lambda}$ vanishes in the neighborhood of a single point $x_{\circ}\in\RE^{n}\backslash\Gamma_{\! 0}$. By \eqref{adjoint3} again, $(\Delta-\lambda) u_{\lambda}=0$ outside some sufficiently large ball $B$ containing $\Omega_{\-}$. Thus, by Rellich's estimate, one gets $u_{\lambda}|B^{c}=0$ (see e.g. \cite[Corollary 4.8]{Leis}) and the proof is done.
\end{proof}
\begin{remark}\label{negative2} Obviously, in the case $\Gamma_{\! 0}=\Gamma$, one has that  $\RE^{n}\backslash\Gamma=\Omega_{\-}\cup\Omega_{\+}$ is not connected. However, if $\Omega_{\+}$ is connected then, by the same kind of reasonings as in the proof of Theorem \ref{negative1}, one gets $u_{\lambda}|\Omega_{\+}=0$. Thus, if the boundary conditions appearing in $\dom(\Delta_{\Pi,\Theta})$ are such that 
$$
\ u|\Omega_{\+}=0\,,\   (\Delta u-\lambda u)|\Omega_{\-}=0\,,\ u\in\dom(\Delta_{\Pi,\Theta}) \quad\Longrightarrow\quad u|\Omega_{\-}=0\,,
$$
then $E^{-}_{\Pi,\Theta}=(-\infty,0)$. For example, two cases where this hypothesis holds are the $\delta$- and $\delta'$-interactions on $\Gamma$  which corresponds to the semi-transparent boundary conditions 
$$[\gamma_{0}]u=0\,,\quad\alpha\gamma_{0}u=[\gamma_{1}]u$$ and 
$$
[\gamma_{1}]u=0\,,\quad
\beta\gamma_{1}u=[\gamma_{0}]u
$$ 
respectively (see Subsections \ref{delta} and \ref{deltaprimo}).
\end{remark}
\end{section}
\begin{section}{The Limiting Absorption Principle}

We begin by recalling the limiting absorption principle  for the self-adjoint operator representing the free Laplacian $\Delta:H^{2}(\RE^{n})\subseteq L^{2}(\RE^{n})\to L^{2}(\RE^{n})$ (see e.g. \cite[Section 4]{Agm}):
\begin{theorem}\label{LAP0} For any $k\in\RE\backslash\{0\}$ and for any $\alpha>{\frac12}$, the  limits 
\be\label{lim1}
R^{\pm}_{-k^{2}}:=\lim_{\epsilon\downarrow 0} (-\Delta-(k^{2}\pm i\epsilon))^{-1}
\ee
exist in $\B(L^{2}_{\alpha}(\RE^{n}),H^{2}_{-\alpha}(
\mathbb{R}^{n}) )$. Setting $\CO_{\pm}:=\{z\in\CO:\pm\text{\rm Im}(z)> 0\}$ and
$$
R^{\pm}_{z}:=\begin{cases} (-\Delta+z)^{-1}\,,&z\in\CO_{\pm}\\
R^{\pm}_{\lambda}\,,&\lambda\in(-\infty,0)\,,
\end{cases}
$$
the maps $z\mapsto R^{\pm}_{z}$ are continuous on $\CO_{\pm}\cup(-\infty,0)$ to $\B(L^{2}_{\alpha}(\RE^{n}),H^{2}_{-\alpha}(\mathbb{R}^{n}) )$.
\end{theorem} 
The existence of the resolvent's limits on the continuous spectrum have been
discussed in \cite{Ren1},\cite{Ren2} for a wide class of operators including
singular perturbations. The general results there provided allow to
prove, in our case, a limiting absorption principle for $\Delta_ { 
\Pi,\Theta  }$: 
\begin{theorem}
\label{Theorem_LAP}Let $\Delta_ {\Pi,\Theta  }$ be defined as in 
Theorem \ref{Theorem_Krein} and assume that it is bounded from above and that \eqref{Theta_reg1} holds true.
Then $(-\infty,0)\cap\sigma_{p}(\Delta_{\Pi,\Theta})$ is a (possibly empty) discrete set of eigenvalues of finite multiplicity and the limits%
\begin{equation}
R_{\Pi,\Theta,-k^{2}}^{\pm} :=\lim_{\epsilon\downarrow 0}\,(  - \Delta_ {(   \Pi,\Theta)
}-(k^{2}\pm i\epsilon))  ^{-1}
\label{LAP}%
\end{equation}
exist in ${\B}(   L^{2}_{\alpha}( 
\mathbb{R}^{n})  ,L^{2}_{-\alpha}(   \mathbb{R}^{n})  )  $
for all $\alpha>{\frac12}$ and for all $k\in\RE\backslash\{0\}$ such that $-k^{2}\in E^{-}_{\Pi,\Theta}$.
\end{theorem}
\begin{proof} Let us at first show that the following four assumptions
hold true for any $z\in\CO\backslash(-\infty,0]$:
\be\label{H1.1}
(-\Delta+z)^{-1}\in\B(L^{2}_{\alpha}(\RE^{n}))\,,
\ee
\be\label{H1.2}
 (-\Delta_ {   \Pi,\Theta}+z)^{-1}\in\B(L^{2}_{\alpha}(\RE^{n}))\,,
\ee
\be\label{H2}
(-\Delta+z)^{-1}-\ (-\Delta_ {   \Pi,\Theta}+z)^{-1}\in{\mathfrak S}_{\infty}(L^{2}(\RE^{n}),L^{2}_{\beta}(\RE^{n}))\,,\quad \beta>2\alpha\,,
\ee
and for all compact subset $K\subset(0,+\infty)$ there exists a constant $c_{K}>0$ such that, for any  $k^{2}\in K$,
\be\label{H3}
\forall u\in L^{2}_{2\alpha}(\RE^{n})\cap\ker(R^{+}_{-k^{2}}-R^{-}_{-k^{2}}),\quad 
\|R^{\pm}_{-k^{2}}u\|_{L^{2}(\RE^{n})}\le c_{K}\|u\|_{L^{2}_{2\alpha}(\RE^{n})}.
\ee
By \cite[Lemma 1, page 170]{R&S}, for all $\sigma\in\RE$ one has 
\be\label{RS}
 (-\Delta+z)^{-1}\in \B(L^{2}_{\sigma}(\RE^{n}))\,.
\ee
Therefore \eqref{H1.1} holds true (see also  \cite[Lemma 5.2 and Remark 5.1]{Ren2}). \par
Introducing the equivalent norm in $H^{2}_{\sigma}(   \mathbb{R}^{n})$
$$
\boldsymbol{|} u\boldsymbol{|} _{H^{2}_{\sigma}(   \mathbb{R}^{n})
}:=\int_{\RE^{n}}\left(1+\|x\|^{2}\right)^{\sigma}|(-\Delta+1)u(x)|^{2}dx\,,
$$
by
\begin{align*}
\boldsymbol{|}(-\Delta+z)^{-1}u\boldsymbol{|}^{2}_{ H^{2}_{\sigma}(\RE^{n})}
=&\|(-\Delta+1)(-\Delta+z)^{-1}u\|^{2}_{L^{2}_{\sigma}(   \mathbb{R}^{n})}\\
\le & 
\|u\|^{2}_{L^{2}_{\sigma}(   \mathbb{R}^{n})}+|1-z|\,\|(-\Delta+z)^{-1}u\|^{2}_{L^{2}_{\sigma}(   \mathbb{R}^{n})}
 \end{align*}
and by \eqref{RS}, one gets 
\be\label{sigma}
(-\Delta+z)^{-1}\in\B(L^{2}_{\sigma}(\RE^{n}), H^{2}_{\sigma}(\RE^{n}))\,.
\ee
Let $\chi\in \C^{\infty}_{\comp} (\RE^{n})
$ such that $\chi|\t\Omega=1$, $\t\Omega\supset\overline \Omega$. Then the map $u\mapsto\chi u$ belongs to $\B(H^{2}_{\sigma}(\RE^{n}),H^{2}(\RE^{n}))$ and so, since $\gamma  (\chi u)=\gamma   u$, by \eqref{sigma} one gets 
\be\label{beta1}\gamma(   -\Delta+z)  ^{-1}\in{\B}(   L^{2}_{\sigma}  
(   \mathbb{R}^{n})  ,H^{{\frac32}}(\Gamma)\oplus H^{{\frac12}}(   \Gamma)  
)  \,.
\ee
and, by duality, 
\be\label{beta2}
G_{z}=(\gamma(   -A+\bar z)  ^{-1})^{\prime}\in{\B}(H^{-{\frac32}}(\Gamma)\oplus H^{-{\frac12}}(   \Gamma),    L^{2}_{-\sigma }(   \mathbb{R}^{n})
) \,. 
\ee
Then, using \eqref{beta1} and \eqref{beta2} with $\sigma=\alpha$ and with $\sigma=-\alpha$ respectively,  \eqref{H1.2} follows from 
\eqref{Krein} and \eqref{RS}. \par Assumption \eqref{Theta_reg1} implies that $\ran((\Theta+\Pi M_{z}\Pi')^{-1})\subseteq H^{s_{1}}(\Gamma)\oplus H^{s_{2}}(\Gamma)$. Thus, by the compact embedding $H^{s_{1}}(\Gamma)\oplus H^{s_{2}}(\Gamma)\hookrightarrow H^{-3/2}(\Gamma)\oplus H^{-1/2}(\Gamma)$, one gets 
\be\label{compact}
(\Theta+\Pi M_{z}\Pi')^{-1}\in{\mathfrak S}_{\infty}(\ran(\Pi),\ran(\Pi'))
\ee
(see the proof of Lemma 4.10 in \cite{MaPoSi} for more details). Therefore, by 
\eqref{Krein}, using \eqref{beta1} and \eqref{beta2} with $\sigma=0$  and $\sigma=-\beta$ respectively, one obtains \eqref{H2}. \par 
Finally, \eqref{H3} holds true by \cite[Corollary 5.7(b)]{BAD}.
\par 
Assumptions \eqref{H1.1}-\eqref{H3} permit us to apply the abstract results provided in \cite{Ren1}: hypothesis (T1) and (E1) in \cite[page 175]{Ren1} corresponds to our \eqref{lim1}, \eqref{H3} and  \eqref{H2} respectively; then, by \cite[Proposition 4.2]{Ren1}, the latters imply hypotheses (LAP) and (E) in \cite[page 166]{Ren1}, i.e. \eqref{lim1} again and  
$$
(-\Delta+z)^{-1}-\ (-\Delta_ {   \Pi,\Theta}+z)^{-1}\in{\mathfrak S}_{\infty}(L^{2}_{-\alpha}(\RE^{n}),L^{2}_{\alpha}(\RE^{n}))\,,
$$
and hypothesis (T) in \cite[page 168]{Ren1}, a technical variant of \eqref{H3}. By \cite[Theorem 3.5]{Ren1}, these last hypotheses, together with the assumption that $-\Delta_{\Pi,\Theta}$ is bounded from below and \eqref{H1.1}-\eqref{H1.2} (i.e. hypothesis (OP) in \cite[page 165]{Ren1}), finally give the content of the theorem.
\end{proof}
\begin{remark}
Since the map defined in \eqref{beta1} has closed range (it is surjective), $G_{z}\in{\B}(H^{-{\frac32}}(\Gamma)\oplus H^{-{\frac12}}(   \Gamma),    L^{2}_{-\sigma }(   \mathbb{R}^{n}))$ defined  in \eqref{beta2} is injective and has closed range by the closed range theorem. Hence, 
by \cite[Theorem 5.2, page 231]{Kato}, for any $z\in\rho(A)$ there exists $c_{z}>0$ such that 
\begin{equation}\label{crt2}
\left\Vert G_{z}(\phi\oplus\varphi)\right\Vert ^{2}_{L^{2}_{-\sigma}( \mathbb{R}^{n})
}\ge c_{z}\left(\|\phi\|^{2}_{H^{-{\frac32}}(  \Gamma)}+\|\varphi\|^{2}_{H^{-{\frac12}}(  \Gamma)}\right)
\end{equation}
for all $\phi\oplus\varphi\in H^{-{\frac32}}(  \Gamma)\oplus H^{-{\frac12}}(  \Gamma)$.
\end{remark}
\begin{lemma}\label{Lemma_G_z} For any $k\in\RE\backslash\{0\}$ and for any $\alpha>{\frac12}$, the  limits 
\be\label{G1}
G^{\pm}_{-k^{2}}:=\lim_{\epsilon\downarrow 0} G_{-(k^{2}\pm i\epsilon)}
\ee
exist in $\B(H^{-{\frac32}}(  \Gamma)\oplus H^{-{\frac12}}(  \Gamma) ,L^{2}_{-\alpha}(
\mathbb{R}^{n}) )$ and
\be\label{G2}
G^{\pm}_{-k^{2}}=G_{z}+(z+k^{2})R^{\pm}_{-k^{2}}G_{z}\,,\quad z\in\CO\backslash(-\infty,0]\,,
\ee
\be\label{G3}
(G^{\pm}_{-k^{2}})^{\prime}=\gamma R^{\mp}_{-k^{2}}\,.
\ee
The function $G^{\pm}_{-k^{2}}(\phi\oplus\varphi)$  solves, in the distribution space ${\mathscr D}'(\RE^{n}\backslash\Gamma)$ and for any $\phi\oplus\varphi\in H^{-{\frac32}}(  \Gamma)\oplus H^{-{\frac12}}(  \Gamma)$, the equation 
$$(\Delta+k^{2})G^{\pm}_{-k^{2}}(\phi\oplus\varphi)=0\,.
$$
Moreover there exist $c^{\pm}_{k^{2}}>0$ such that
\begin{equation}
\left\Vert G_{-k^{2}}^{\pm }(\phi\oplus\varphi)\right\Vert ^{2}_{L^{2}_{-\alpha}( \mathbb{R}^{n})
}\ge c_{k^{2}}^{\pm}\left(\|\phi\|^{2}_{H^{-{\frac32}}(  \Gamma)}+\|\varphi\|^{2}_{H^{-{\frac12}}(  \Gamma)}\right)\,.
\label{G_z_bijection}%
\end{equation}
\end{lemma}
\begin{proof} Let $\chi\in C^{\infty}_{\comp} (\RE^{n})
$ such that $\chi|\t\Omega=1$, $\t\Omega\supset\overline \Omega$. Then the map $u\mapsto\chi u$ belongs to $\B(H^{2}_{-\alpha}(\RE^{n}),H^{2}(\RE^{n}))$ and so, by Theorem \ref{LAP0}, 
the  limits 
\be\label{lim3}
\gamma R^{\pm}_{-k^{2}}=\lim_{\epsilon\downarrow 0} \gamma(-\Delta-(k^{2}\pm i\epsilon))^{-1}
\ee
exist in $\B(L^{2}_{\alpha}(\RE^{n}),H^{{\frac32}}(\Gamma)\oplus H^{{\frac12}}(\Gamma))$. Then the relations 
\begin{align*}
&G_{-(k^{2}\pm i\epsilon)}=(\gamma(-\Delta-k^{2}\mp i\epsilon))^{\prime}\\
=&G_{z}+(z+k^{2}\mp i\epsilon) (-\Delta-(k^{2}\pm i\epsilon))^{-1}G_{z}
\end{align*}
(see \cite[Lemma 2.1]{Posi}), and $G_{z}\in \B(H^{-{\frac12}}(\Gamma)\oplus H^{-{\frac32}}(\Gamma), L^{2}_{\alpha}(\RE))$ (use \eqref{beta2} with $\sigma=-\alpha$), give \eqref{G1}, \eqref{G2} and \eqref{G3}.\par Since 
$$
\lim_{\epsilon\downarrow 0} \|(\Delta+(k^{2}\pm i\epsilon))u-(\Delta+k^{2})u\|_{L^{2}_{\alpha}(\RE^{n})}=0\,,\quad u\in H^{2}_{\alpha}(\RE^{n})\,,
$$
one has
\be\label{inv}
R^{\pm}_{-k^{2}}(-\Delta-k^{2})u=u\,,\quad u\in H^{2}_{\alpha}(\RE^{n})\,.
\ee
Thus, for any test function $u\in{\mathscr D}(\RE^{n}\backslash\Gamma)\equiv C^{\infty}_{\comp} (\RE^{n}\backslash\Gamma)\subset H^{2}_{\alpha}(   \mathbb{R}
^{n}) $, 
one obtains
\begin{align*} 
&\langle (\Delta+k^{2})G^{\pm}_{-k^{2}}(\phi\oplus\varphi),u\rangle=
\langle G^{\pm}_{-k^{2}}(\phi\oplus\varphi),(\Delta+k^{2})u\rangle\\
=&\langle (\gamma  R^{\mp}_{-k^{2}})^{\prime}(\phi\oplus\varphi),(\Delta+k^{2})u\rangle
=-\langle \phi\oplus\varphi,\gamma  R^{\pm}_{-k^{2}}(\Delta+k^{2})u\rangle\\
=&-
\langle \phi\oplus\varphi,\gamma   u\rangle=0\,.
\end{align*}
By \eqref{inv} and the surjectivity of $\gamma  : H^{2}(\RE^{n})\to  H^{{\frac32}}(  \Gamma)\oplus H^{{\frac12}}(  \Gamma)$, the map
$$\gamma  R^{\pm}_{-k^{2}}:L^{2}_{\alpha}(\RE^{n})\to H^{{\frac32}}(  \Gamma)\oplus H^{{\frac12}}(  \Gamma)$$
is surjective: given $\phi\oplus\varphi\in H^{{\frac32}}(  \Gamma)\oplus H^{{\frac12}}(  \Gamma)$ one has 
$\gamma  R^{\pm}_{-k^{2}}v=\phi\oplus\varphi$, where $v=-(\Delta+k^{2})\chi u$ and $\gamma   \chi u=\gamma   u=\phi\oplus\varphi$, $u\in H^{2}(\RE^{n})$.
Therefore, by the closed range theorem, the range of $G_{-k^{2}}^{\pm }$ is closed; since $G_{-k^{2}}^{\pm }$ is injective (it is the dual of a surjective map), \cite[Theorem 5.2, page 231]{Kato} gives \eqref{G_z_bijection}.
\end{proof}
While interesting, Theorem \ref{Theorem_LAP} gives no answer to the obvious question: ''does Kre\u\i n's formula survive in the limit $\epsilon\downarrow 0$?''  That is given by the following
\begin{theorem}
\label{teo_LAP} Under the assumptions of Theorem \ref{Theorem_LAP}, for any $-k^{2}\in E^{-}_{\Pi,\Theta}$,
 the limits 
\be\label{M-lim}
M^{\pm}_{-k^{2}}:=\lim_{\epsilon\downarrow 0} M_{-(k^{2}\pm i\epsilon)}\,,
\ee
\be
L_{\Pi,\Theta,-k^{2}}^{\pm}:=\lim_{\epsilon\downarrow 0} (\Theta+\Pi M_{-(k^{2}\pm i\epsilon)}\Pi')^{-1}
\ee
exist in $\B(H^{-{\frac32}}(\Gamma)\oplus H^{-{\frac12}}(\Gamma), H^{{\frac32}}(\Gamma)\oplus H^{{\frac12}}(\Gamma))$ and $\B(\ran(\Pi),\ran(\Pi'))$ respectively and  
\be\label{M-lim-1}
M_{-k^{2}}^{\pm}=M_{z}-(z+k^{2}) \gamma R^{\pm}_{-k^{2}}G_{z}\,,\quad z\in\CO\backslash(-\infty,0]\,.
\ee
The linear operator $\Theta+\Pi M_{-k^{2}}^{\pm}\Pi'$ has a bounded inverse
\be\label{L-lim-1}
L_{\Pi,\Theta,-k^{2}}^{\pm}=(\Theta+\Pi M_{-k^{2}}^{\pm}\Pi')^{-1}
\ee
and
\be\label{K-lim}
R_{\Pi,\Theta,-k^{2}}^{\pm}-R^{\pm}_{-k^{2}}=G^{\pm}_{-k^{2}}\Pi'
(\Theta+\Pi M_{-k^{2}}^{\pm}\Pi')^{-1}\Pi\gamma R^{\pm}_{-k^{2}}\,.
\ee
Moreover the map $$z\mapsto R^{\pm}_{\Pi,\Theta,z}:=
\begin{cases}(-\Delta_{\Pi,\Theta}+{z})^{-1}\,,&z\in\CO_{\pm}\\
R_{\Pi,\Theta,\lambda}^{\pm}\,,&\lambda\in E^{-}_{\Pi,\Theta}
\end{cases}
$$ is continuous on $\CO_{\pm}\cup E^{-}_{\Pi,\Theta}$ to $\B(L^{2}_{\alpha}(\RE^{n}), L^{2}_{-\alpha}(\RE^{n}))$.
\end{theorem}
\begin{proof} 
By \cite[equation (5)]{Posi}, the operator family $M_{z}$, $z\in\CO\backslash(-\infty,0]$,  satisfies the identity
$$
M_{-(k^{2}\pm i\epsilon)}=M_{z}-(z+k^{2}\mp i\epsilon)\gamma(-\Delta-(k^{2}\pm i\epsilon))^{-1}G_{z}\,.
$$
Since, by \eqref{beta2}, $G_{z}\in \B(H^{-{\frac12}}(\Gamma)\oplus H^{-{\frac32}}(\Gamma), L^{2}_{\alpha}(\RE^{n}))$, the norm convergence of $M_{-(k^{2}\pm i\epsilon)}$ to $M^{\pm}_{-k^{2}}$ is consequence of Theorem \ref{LAP0}. This gives \eqref{M-lim} and \eqref{M-lim-1}.\par
By Theorem \ref{Theorem_LAP}, the limits%
\begin{equation}\label{limG}
\lim_{\epsilon\downarrow 0}G_{-(k^{2}\pm i\epsilon)}\Pi^{\prime}(   \Theta+\Pi M_{-(k^{2}\pm i\epsilon)}%
\Pi^{\prime})  ^{-1}\Pi\gamma(-\Delta-(k^{2}\pm i\epsilon))^{-1}
\end{equation}
exist in $\B(L^2_{\alpha}(   \mathbb{R}^{n}),L^2_{-\alpha}(   \mathbb{R}^{n}))  $ and, by Lemma \ref{Lemma_G_z}, the limits 
\begin{equation}\label{limG1}
G^{\pm}_{-k^{2}}=\lim_{\epsilon\downarrow 0}G_{-(k^{2}\pm i\epsilon)}
\end{equation}
and
\begin{equation}\label{limG2}
(G_{-k^{2}}^\mp)^{\prime}=\gamma R^{\pm}_{-k^{2}}=\lim_{\epsilon\downarrow 0}\gamma(-\Delta-(k^{2}\pm i\epsilon))^{-1}
\end{equation}
exist in $\B(H^{-{\frac12}}(\Gamma)\oplus H^{-{\frac32}}(\Gamma), L^{2}_{-\alpha}(\RE^{n}))$ and
$\B(L^{2}_{\alpha}(\RE^{n}), H^{{\frac12}}(\Gamma)\oplus H^{{\frac32}}(\Gamma))$ respectively.
According to \eqref{crt2} and \eqref{G_z_bijection}, there exist $\t c^{\pm}_{k^{2}}>0$ such that, for all $\epsilon>0$,   
\begin{align*}
&\left\Vert G_{-(k^{2}\pm i\epsilon)}\Pi^{\prime}(\Theta+\Pi M_{-(k^{2}\pm i\epsilon)}\Pi^{\prime})^{-1}\Pi\gamma(-\Delta-(k^{2}\pm i\epsilon))^{-1}u\right\Vert ^{2}_{L^{2}_{-\alpha}( \mathbb{R}^{n})}\\
\ge &\t c_{k^{2}}^{\pm}\left(\|(\Theta+\Pi M_{-(k^{2}\pm i\epsilon)}\Pi^{\prime})^{-1}\Pi\gamma(-\Delta-(k^{2}\pm i\epsilon))^{-1}u\|^{2}_{H^{-{\frac32}}(  \Gamma)\oplus H^{-{\frac12}}(  \Gamma)}\right)\,.
\end{align*}
Let $|\!|\!|\cdot |\!|\!|$ denote the operator norm in $\B(X,Y)$, the Hilbert spaces $X$ and $Y$ varying according to the case. Then, by \eqref{limG}, one has 
$$
\sup_{\epsilon>0}\n (\Theta+\Pi M_{-(k^{2}\pm i\epsilon)}\Pi^{\prime})^{-1}\Pi\gamma(-\Delta-(k^{2}\pm i\epsilon))^{-1}\n<+\infty\,.
$$
and, by duality
$$
\sup_{\epsilon>0}\n G_{-(k^{2}\pm i\epsilon)}\Pi^{\prime}(\Theta+\Pi M_{-(k^{2}\pm i\epsilon)}\Pi^{\prime})^{-1}\n<+\infty\,.
$$
Thus, by 
\begin{align*}
&\n G_{-(k^{2}\pm i\epsilon)}\Pi^{\prime}(\Theta+\Pi M_{-(k^{2}\pm i\epsilon)}\Pi^{\prime})^{-1}\Pi(\gamma(-\Delta-(k^{2}\pm i\epsilon))^{-1}-\gamma R^{\pm}_{-k^{2}})\n\\
\le &\left(\sup_{\epsilon>0}\n G_{-(k^{2}\pm i\epsilon)}\Pi^{\prime}(\Theta+\Pi M_{-(k^{2}\pm i\epsilon)}\Pi^{\prime})^{-1}\n\right)\\
&\times\n \gamma(-\Delta-(k^{2}\pm i\epsilon))^{-1}-\gamma R^{\pm}_{-k^{2}} \n
\end{align*}
and by \eqref{limG},\eqref{limG2}, one has that the limits%
\begin{equation}\label{limG3}
\lim_{\epsilon\downarrow 0}G_{-(k^{2}\pm i\epsilon)}\Pi^{\prime}(   \Theta+\Pi M_{-(k^{2}\pm i\epsilon)}%
\Pi^{\prime})  ^{-1}\Pi\gamma R^{\pm}_{-k^{2}}
\end{equation}
exist in $\B(L^2_{\alpha}(   \mathbb{R}^{n}),L^2_{-\alpha}(   \mathbb{R}^{n}))  $ and coincide with the ones given in \eqref{limG}.  Since the map 
$\gamma R^{\pm}_{-k^{2}}$ is surjective (see the end of the proof of Lemma \ref{Lemma_G_z}), by \cite[Theorem 5.2, page 231]{Kato} there exists $\hat c_{-k^{2}}^{\pm}>0$ such that 
$$
\forall u\in \ker(\gamma R^{\pm}_{-k^{2}})^{\perp}\,,\quad \|\gamma R^{\pm}_{-k^{2}}u\|_{H^{{\frac12}}(\Gamma)\oplus H^{{\frac32}}(\Gamma)}\ge \hat c_{-k^{2}}^{\pm}\|u\|_{L^{2}_{\alpha}(\RE^{n})}\,.
$$
Setting $\Phi_{\epsilon}:=G_{-(k^{2}\pm i\epsilon)}\Pi^{\prime}(   \Theta+\Pi M_{-(k^{2}\pm i\epsilon)}\Pi^{\prime})^{-1}\Pi$, we have 
\begin{align*}
\sup_{\{0\}\not=\phi\oplus\varphi\in H^{{\frac12}}(\Gamma)\oplus H^{{\frac32}}(\Gamma)}\frac{
\|(\Phi_{\epsilon_{1}}-\Phi_{\epsilon_{2}})\phi\oplus\varphi\|_{L^{2}_{-\alpha}(\RE^{n})}}{\|\phi\oplus\varphi\|_{H^{{\frac12}}(\Gamma)\oplus H^{{\frac32}}(\Gamma)}}&\\
=\sup_{u\in\ker(\gamma R^{\pm}_{-k^{2}})^{\perp}}\frac{
\|(\Phi_{\epsilon_{1}}-\Phi_{\epsilon_{2}})\gamma R^{\pm}_{-k^{2}}u\|_{L^{2}_{-\alpha}(\RE^{n})}}{\|\gamma R^{\pm}_{-k^{2}}u\|
_{H^{{\frac12}}(\Gamma)\oplus H^{{\frac32}}(\Gamma)}}&\\
\le \sup_{u\in\ker(\gamma R^{\pm}_{-k^{2}})^{\perp}}\frac{
\|(\Phi_{\epsilon_{1}}-\Phi_{\epsilon_{2}})\gamma R^{\pm}_{-k^{2}}u\|_{L^{2}_{-\alpha}(\RE^{n})}}{\hat c_{-k^{2}}^{\pm}\|u\|_{L^{2}_{\alpha}(\RE^{n})}}&\,.
\end{align*}
Hence, by \eqref{limG3}, the limits 
\begin{equation}\label{limG4}
\lim_{\epsilon\downarrow 0}G_{-(k^{2}\pm i\epsilon)}\Pi^{\prime}(   \Theta+\Pi M_{-(k^{2}\pm i\epsilon)}%
\Pi^{\prime})  ^{-1}\Pi
\end{equation}
exist in $\B(\ran(\Pi),L^{2}_{-\alpha}(\RE^{n}))$ and, by duality, the limits 
\begin{equation}\label{limG5}
\lim_{\epsilon\downarrow 0}(   \Theta+\Pi M_{-(k^{2}\pm i\epsilon)}%
\Pi^{\prime})  ^{-1}\Pi\gamma (-\Delta-(k^{2}\pm i\epsilon))^{-1}
\end{equation}
exist in $\B(L^{2}_{\alpha}(\RE^{n}), \ran(\Pi'))$. By \eqref{crt2} and \eqref{G_z_bijection}, 
it results
\begin{align*}
&\left\Vert G_{-(k^{2}\pm i\epsilon)}\Pi^{\prime}(\Theta+\Pi M_{-(k^{2}\pm i\epsilon)}\Pi^{\prime})^{-1}\Pi(\phi\oplus\varphi)\right\Vert ^{2}_{L^{2}_{-\alpha}( \mathbb{R}^{n})}\\
\ge &\t c_{k^{2}}^{\pm}\left(\|(\Theta+\Pi M_{-(k^{2}\pm i\epsilon)}\Pi^{\prime})^{-1}\Pi(\phi\oplus\varphi)\|^{2}_{H^{-{\frac32}}(  \Gamma)\oplus H^{-{\frac12}}(  \Gamma)}\right)
\end{align*}
and so  \eqref{limG4} gives
\be\label{sup-1}
\sup_{\epsilon>0}\n (\Theta+\Pi M_{-(k^{2}\pm i\epsilon)}
\Pi^{\prime})^{-1}\n<+\infty\,.
\ee
Therefore, by \eqref{limG5}, one gets that the limits 
\begin{equation}\label{limG6}
\lim_{\epsilon\downarrow 0}(   \Theta+\Pi M_{-(k^{2}\pm i\epsilon)}%
\Pi^{\prime})  ^{-1}\Pi \gamma R^{\pm}_{-k^{2}}
\end{equation}
exist in $\B(L^{2}_{\alpha}(\RE^{n}), \ran(\Pi'))$ and coincide with the ones given by \eqref{limG5}. 
Since $\gamma R^{\pm}_{-k^{2}}$ is surjective, proceeding as above one gets the existence of the 
limits
\begin{equation}\label{limG7}
\lim_{\epsilon\downarrow 0}(   \Theta+\Pi M_{-(k^{2}\pm i\epsilon)}%
\Pi^{\prime})  ^{-1}
\end{equation}
with respect to the operator norm in $\B(\ran(\Pi), \ran(\Pi'))$. Finally, taking the limit $\epsilon\downarrow 0$ in the identities
\begin{align*}
&(\Theta+\Pi M_{-(k^{2}\pm i\epsilon)}\Pi^{\prime})^{-1}(\Theta+\Pi M_{-(k^{2}\pm i\epsilon)}\Pi^{\prime})\\=&
(\Theta+\Pi M_{-(k^{2}\pm i\epsilon)}\Pi^{\prime})(\Theta+\Pi M_{-(k^{2}\pm i\epsilon)}\Pi^{\prime})^{-1}=\uno\,,
\end{align*}
one gets
$$
(\Theta+\Pi M_{-k^{2}}^\pm\Pi^{\prime})^{-1}(\Theta+\Pi M_{-k^{2}}^\pm\Pi^{\prime})=
(\Theta+\Pi M_{-k^{2}}^\pm\Pi^{\prime})(\Theta+\Pi M_{-k^{2}}^\pm\Pi^{\prime})^{-1}=\uno\,.
$$
By Theorem \ref{LAP0}, the map $z\mapsto \gamma R^{\pm}_{z}$ is continuous on $\CO_{\pm}\cup(-\infty,0)$ to $\B(L^{2}_{\alpha}(\RE^{n}), H^{{\frac32}}(\Gamma)\oplus H^{{\frac12}}(\Gamma))$;  by duality, $z\mapsto G^{\pm}_{z}=(\gamma R^{\pm}_{z})^{\prime}$ is continuos on $\CO_{\pm}\cup(-\infty,0)$ to $\B(H^{-{\frac32}}(\Gamma)\oplus H^{-{\frac12}}(\Gamma),L^{2}_{-\alpha}(\RE^{n}))$.
By \eqref{M-lim-1}, $z\mapsto M^{\pm}_{z}$ is continuos on $\CO_{\pm}\cup(-\infty,0)$ to $\B(H^{{\frac32}}(\Gamma)\oplus H^{{\frac12}}(\Gamma),H^{-{\frac32}}(\Gamma)\oplus H^{-{\frac12}}(\Gamma))$ and so, by \eqref{sup-1}, $z\mapsto (\Theta+\Pi M_{z}^\pm\Pi^{\prime})^{-1}$ is continuos on $\CO_{\pm}\cup(-\infty,0)$ to $\B(\ran(\Pi),\ran(\Pi'))$. In conclusion $z\mapsto  R^{\pm}_{\Pi,\Theta,z}$ is continuous on $\CO_{\pm}\cup(-\infty,0)$ to $\B(L^{2}_{\alpha}(\RE^{n}), L^{2}_{-\alpha}(\RE^{n}))$.
\end{proof}
\begin{corollary}\label{rem} If, in addition to hypotheses inTheorem \ref{teo_LAP}, hypothesis \eqref{Theta_reg} holds, then the Kre\u\i n type formula \eqref{K-lim} can be re-written as 
\be\label{krein-lim}
R_{\Pi,\Theta,-k^{2}}^{\pm}-R^{\pm}_{-k^{2}}=G^{\pm}_{-k^{2}}\Pi'
(B_\Theta-\Pi \gamma G^{\pm}_{-k^{2}}\Pi^{\prime})^{-1}\Pi\gamma R^{\pm}_{-k^{2}}\,.
\ee
\end{corollary}
\begin{proof}
By hypothesis \eqref{Theta_reg}, one has  $$\ran(G_{z})|\dom(\Theta)\subseteq 
H^{2}(\RE^{n}\backslash\Gamma)\,.
$$
Then, by \eqref{G2} and \eqref{lim3}, the operator 
$$ 
\gamma G^{\pm }_{-k^{2}}:\dom(\Theta)\to H^{{\frac32}}(\Gamma)\oplus H^{{\frac12}}(\Gamma)
$$
is well-defined and, for any $\phi\oplus\varphi\in\dom(\Theta)$, the limits
\be\label{gammaG-lim}
\gamma G^{\pm }_{-k^{2}}(\phi\oplus\varphi)=\lim_{\epsilon\downarrow 0} 
\gamma G_{-(k^{2}\pm i\epsilon)}(\phi\oplus\varphi)
\ee
exists in $H^{{\frac32}}(\Gamma)\oplus H^{{\frac12}}(\Gamma)$. Thus, for any $\phi\oplus\varphi\in\dom(\Theta)$ and for any $-k^{2}\in E^{-}_{\Pi,\Theta}$,  one has 
\begin{align*}
&(\Theta+\Pi M^{\pm}_{-k^{2}}\Pi^{\prime})(\phi\oplus\varphi)=
\lim_{\epsilon\downarrow 0}(\Theta+\Pi M_{-(k^{2}\pm i\epsilon)}\Pi^{\prime})(\phi\oplus\varphi)\\
=&\lim_{\epsilon\downarrow 0}(B_\Theta-\Pi \gamma G_{-(k^{2}\pm i\epsilon)}\Pi^{\prime})(\phi\oplus\varphi)=(B_\Theta-\Pi \gamma G^{\pm}_{-k^{2}}\Pi^{\prime})(\phi\oplus\varphi)\,.
\end{align*}
\end{proof}
\begin{corollary}\label{sing} Under the hypotheses in Theorem \ref{Theorem_LAP}, 
$\Delta_{\Pi,\Theta}$ has empty singular continuous spectrum, i.e.
$$
L^{2}(\RE^{n})=L^{2}(\RE^{n})_{ac}\oplus L^{2}(\RE^{n})_{pp}\,,
$$
where $L^{2}(\RE^{n})_{ac}$ and $L^{2}(\RE^{n})_{pp}$ denote the absolutely continuous and pure point subspaces of $L^{2}(\RE^{n})$ with respect to $\Delta_{\Pi,\Theta}$.
\end{corollary}
\begin{proof} We follow standard arguments (see e.g. \cite[Theorem 6.1]{Agm}): let $E_{\lambda}$ the spectral resolution of $\Delta_{\Pi,\Theta}$ and 
let $u\in L^{2}_{\alpha}(\RE^{n})\cap L^{2}(\RE^{n})_{pp}^{\perp}$. Then, for any compact interval 
$[a,b]\subset E^{-}_{\Pi,\Theta}$ one has, by Stone's formula, by the continuity of $z\mapsto R^{\pm}_{\Pi,\Theta,z}$ and by Lebesgue's dominated converge theorem,  
\begin{align*}
&\langle (E_{b}-E_{a})u,u\rangle_{L^{2}(\RE^{n})}\\
=&\lim_{\epsilon\downarrow 0}\frac1{2\pi i}\int_{a}^{b}\langle((-\Delta_{\Pi,\Theta}+\lambda-i\epsilon)^{-1}-(-\Delta_{\Pi,\Theta}+\lambda+i\epsilon)^{-1})u,u\rangle\,d\lambda\\
=&\frac1{2\pi i}\int_{a}^{b}\langle(R^{-}_{\Pi,\Theta,\lambda}-R^{+}_{\Pi,\Theta,\lambda})u,u\rangle\,d\lambda
\end{align*}
so that $\langle E_{\lambda}u,u\rangle_{L^{2}(\RE^{n})}$ is differentiable on $ E^{-}_{\Pi,\Theta}$ and 
$$\frac{d\ }{d\lambda}\langle E_{\lambda}u,u\rangle_{L^{2}(\RE^{n})}=\frac1{2\pi i}\,\langle(R^{-}_{\Pi,\Theta,\lambda}-R^{+}_{\Pi,\Theta,\lambda})u,u\rangle
$$
for all $u\in L^{2}_{\alpha}(\RE^{n})$. Since it is known that the set of functions for which $\langle E_{\lambda}u,u\rangle_{L^{2}(\RE^{n})}$ is differentiable is a closed set, in conclusion $\langle E_{\lambda}u,u\rangle_{L^{2}(\RE^{n})}$ is differentiable for any $u\in L^{2}(\RE^{n})_{pp}^{\perp}$.
\end{proof}

\end{section}
\begin{section}{Eigenfunction expansion and the scattering matrix}
All over this section we suppose that the assumptions in Theorem \ref{Theorem_LAP} hold true. We then  consider the extension $\t\Delta_{\Pi,\Theta}$ of the self-adjoint operator 
$$
\Delta_{\Pi,\Theta}:\dom(\Delta_{\Pi,\Theta})\subseteq L^{2}(\RE^{n})\to L^{2}(\RE^{n})
$$
to the larger space $L^{2}_{-\alpha}(\RE^{n})$, $\alpha>0$, given by 
$$
\t\Delta_{\Pi,\Theta}:\dom(\t\Delta_{\Pi,\Theta})\subseteq L^{2}_{-\alpha}(\RE^{n})\to L^{2}_{-\alpha}(\RE^{n})\,,\quad 
$$
$$
\t\Delta_{\Pi,\Theta}u:=\Delta u_{\circ}+G(\phi\oplus\varphi)=
\Delta u-[\hat \gamma_{1}]u\,\delta_\Gamma-[\hat \gamma_{0}]u\,\nu\!\cdot\!\nabla\delta_\Gamma\,,
$$ 
\begin{align*}
&\dom(\t\Delta_ {   \Pi,\Theta  })\\
:=&\{u=u_{\circ}+G(\phi\oplus\varphi):u_{\circ}\in
H^{2}_{-\alpha}(\mathbb{R}^{n}),\, \phi\oplus\varphi\in \dom(   \Theta),\, \Pi\gamma  
u_{\circ}=\Theta(\phi\oplus\varphi)\}.
\end{align*}
By
Theorem \ref{Theorem_LAP} and \eqref{adjoint3} one has 
$$
\text{\rm graph}(\t\Delta_ {   \Pi,\Theta  })\cap (L^{2}(\RE^{n})\oplus L^{2}(\RE^{n}))
=\text{\rm graph}(\Delta_ {   \Pi,\Theta  })\,.
$$ 
Such an operator $\t \Delta_{\Pi,\Theta}$ allows the introduction of the generalized eigenfunctions of $\Delta_{\Pi,\Theta}$:
\begin{theorem}\label{eigenfunctions} Let $\t u_{k}\in L^{2}_{-\alpha}(\RE^{n})\backslash\{0\}$ be a generalized eigenfunction of $\Delta_{\Pi,\Theta}$ with eigenvalue $-k^{2}\in E^{-}_{\Pi,\Theta}$, i.e. $\t u_{k}$ belongs to $\dom(\t\Delta_{\Pi,\Theta})$ and solves the equation $$(\t\Delta_{\Pi,\Theta}+k^{2})\t u_{k}=0\,.$$  Then 
$$
\t u_{k}= u_{k}+G_{-k^{2}}^\pm\Pi^{\prime}(  \Theta+\Pi M_{-k^{2}}^{\pm}\Pi^{\prime})  ^{-1}\Pi\gamma u_{k}\,,
$$
where $u_{k}\in H^{2}_{-\alpha}(\RE^{n})$ is a generalized eigenfunction of 
$\Delta:H^{2}(\RE^{n})\subset L^{2}(\RE^{n})\to L^{2}(\RE^{n})$ with eigenvalue $-k^{2}$. 
\end{theorem}
\begin{proof} Let us set $\t u_{k}=u_{\circ}+G(\phi\oplus\varphi)$, with $u_{\circ}\in
H^{2}_{-\alpha}(\mathbb{R}^{n})$ and $\phi\oplus\varphi\in \dom(   \Theta)$ such that  $\Pi\gamma  u_{\circ}=\Theta(\phi\oplus\varphi)$. Then $(\t\Delta_{\Pi,\Theta}+k^{2})\t u_{k}=0$ gives 
$$
(\Delta+k^{2})u_{\circ}=-(1+k^{2})G(\phi\oplus\varphi)\,.
$$ 
Since $\ran (G)\subset L^{2}_{\alpha}(\RE^{n})$, we can apply $R^{\pm}_{-k^{2}}$ to both sides of the above relation; thus 
$$
u_{\circ}=u_{k}+(1+k^{2})R^{\pm}_{-k^{2}}G(\phi\oplus\varphi)\,,
$$ 
where $u_{k}\in H^{2}_{-\alpha}(\RE^{n})$ is any solution of the equation $(\Delta+k^{2})u_{k}=0$. Imposing the boundary conditions we obtain, by \eqref{M-lim-1} and by $M_{1}=0$,
$$
\Pi\gamma u_\circ=\Pi\gamma u_{k}-\Pi M^{\pm}_{-k^{2}}(\phi\oplus\varphi)=\Theta(\phi\oplus\varphi)\,,
$$
i.e. $$\phi\oplus\varphi=(\Theta+\Pi M^{\pm}_{-k^{2}}\Pi')^{-1}\Pi\gamma u_{k}\,.$$
The proof is then concluded by using \eqref{G2} with $z=1$.
\end{proof}
\begin{remark} 
Under hypothesis \eqref{Theta_reg}, by Corollary \ref{Lemma_parameter}, one can alternatively define 
$$
\dom(\t\Delta_ {   \Pi,\Theta  }):=\{u\in L^{2}_{-\alpha}(\RE^{n})\cap H^{2}_{-\alpha}(\RE^{n}\backslash\Gamma):[\gamma]u\in\dom(\Theta)\,,\  \Pi\gamma  u=B_\Theta[\gamma]u\}
$$
and so
$$
\t u_{k}= u_{k}+G_{-k^{2}}^\pm\Pi^{\prime}(  B_\Theta-\Pi \gamma G_{-k^{2}}^{\pm}\Pi^{\prime})  ^{-1}\Pi\gamma u_{k}\,.
$$
\end{remark}
Before stating the next results we recall the following definition: let $u$ solve  the Helmholtz equation $(\Delta+k^{2})u=0$ on  the exterior of some bounded domain; we say that $u$ satisfies the $(\pm)$ {\it Sommerfeld radiation condition} whenever 
$$
\lim_{\|x\|\to+\infty}\,\|x\|^{(n-1)/2}(\hat x\!\cdot\!\nabla \pm ik)u(x)=0
$$
hold uniformly in $\hat x:=x/\|x\|$. The plus sign corresponds to an inward wave and the minus one corresponds to a outward wave. 
\begin{lemma}\label{SRC} 1) 
the functions $G^{\pm}_{-k^{2}}(\phi\oplus\varphi)$ satisfy the $(\pm)$ Sommerfeld radiation condition.\par\noindent
2) If $u\in \ker(\t\Delta_{\Pi,\Theta}+k^{2})$, $-k^{2}\in E^{-}_{\Pi,\Theta}$, 
satisfies the Sommerfeld radiation condition, then $u=0$.
\end{lemma}
\begin{proof} By \eqref{G3},
$$
G^{\pm}_{-k^{2}}(\phi\oplus\varphi)=\SL^{\pm}_{-k^2}\phi+\DL^{\pm}_{-k^2}\varphi\,,
$$
where
$$
\SL^{\pm}_{-k^2}:=(\gamma_{0}R^{\mp}_{-k^{2}})^{\prime}\,,\quad 
\DL^{\pm}_{-k^2}:=(\gamma_{1}R^{\mp}_{-k^{2}})^{\prime}\,.
$$
By \eqref{green} and $K_{\alpha}(z)=\frac{\pi}2\,i^{\alpha+1}H^{(1)}_{\alpha}(iz)$, where $H^{(1)}_{\alpha}$ denotes the Hankel function of first kind of order $\alpha$, it results
\begin{align*}
\G_{-(k^{2}\pm i\epsilon)}(x)=\frac{i}{4}\,\left(\frac{\sqrt {k^{2}\mp i\epsilon}}{2\pi\|x\|}\,\right)^{\frac{n}2-1}\!\!\!\!H^{(1)}_{\frac{n}2-1}(\sqrt{k^{2}\mp i\epsilon}\,\|x\|)\,,\quad
 \text{Im}\sqrt{k^{2}\mp i\epsilon}>0\,.
\end{align*}
Therefore, by $\sqrt {k^{2}\mp i\epsilon}=\mp |k|+\frac{i}{2}\,\frac{\epsilon}{|k|}+o(\epsilon)$ for any $k\not=0$,
\begin{align*}
\G^{\pm}_{-k^{2}}(x):=\lim_{\epsilon\downarrow 0}\G_{-(k^{2}\pm i\epsilon)}(x)=\frac{i}{4}\,\left(\frac{\mp |k|}{2\pi\|x\|}\,\right)^{\frac{n}2-1}\!\!\!\!H^{(1)}_{\frac{n}2-1}(\mp|k|\,\|x\|)\,.
\end{align*}
Thus, for any fixed $k\not=0$, one gets (see e.g. \cite[Appendix 1]{AS})
$$
\sup_{\epsilon>0}|\G_{-(k^{2}\pm i\epsilon)}(x)|\le c\left({\|x\|^{-\frac12(n-1)}}+{\|x\|^{-(n-2)}} \right)\,.
$$
By the dominated convergence theorem, for any bounded and compactly supported $u$ and $v$,
$$\lim_{\epsilon\downarrow 0}\langle u,(R_{-(k^{2}\pm i\epsilon)}-\t R^{\pm}_{-k^{2}})v\rangle_{L^{2}(\RE^{n})}=0\,,
$$
where $\t R^{\pm}_{-k^{2}}$ denotes the operator with integral kernel given by $\G^{\pm}_{-k^{2}}(x-y)$. Therefore, by Theorem \ref{LAP0}, $R^{\pm}_{-k^{2}}=\t R^{\pm}_{-k^{2}}$ and so, if $\phi$ and $\varphi$ are in $L^{2}(\Gamma)$ and $x\notin\Gamma$, 
\begin{equation}\label{K1}
\SL^{\pm}_{-k^{2}}\phi(x)=\int_{\Gamma}\G_{-k^{2}}^{\pm}(x-y)\,\phi(y)\,d\sigma (y)
\end{equation}
and%
\begin{equation}\label{K2}
\DL^{\pm}_{-k^{2}} \varphi(x)=\int_{\Gamma}\nu(y)\!\cdot\!\nabla \G_{-k^{2}}^{\pm}(x-y)\,\varphi(y)\,d\sigma (y)\,.
\end{equation}
Then, by the behavior of $H^{(1)}_{\alpha}(x)$ and $\frac{d\,}{dx}H^{(1)}_{\alpha}(x)$ as $\|x\|\to+\infty$, there follows that  $\SL^{\pm}_{-k^{2}}\phi=\G_{-k^{2}}^{\pm}*(\phi\,\delta_{\Gamma})$ and 
$\DL^{\pm}_{-k^{2}}\varphi=\G_{-k^{2}}^{\pm}*(\varphi\,\nu\!\cdot\!\nabla\delta_{\Gamma})$ (and hence  $G^{\pm}_{-k^{2}}(\phi\oplus\varphi)$) satisfy the $(\pm)$ Sommerfeld radiation condition (see e.g. \cite[Lemma 7, Subsection 7d, Section 8, Chapter II]{DL}).
\par\noindent
2) Let us suppose that $u\not=0$. 
Then, by Theorem \ref{eigenfunctions},  
$$
u=u_{k}+G_{-k^{2}}^\pm\Pi^{\prime}(  \Theta+\Pi M_{-k^{2}}^{\pm}\Pi^{\prime})  ^{-1}\Pi\gamma u_{k}\,,
$$
where $u_{k}\in H^{2}_{-\alpha}(\RE^{n})$ is a generalized eigenfunction of 
$\Delta$ with eigenvalue $-k^{2}$. Then, by 1), $u$ satisfies the Sommerfeld radiation condition if and only if $u_{k}$ does. By Green's formula on the ball of radius $R$, since $(\Delta +k^{2})u_{k}(x)=0$ for any $x\in\RE^{3}$, one has 
$$\text{\rm Im}\left(\int_{\|x\|=R}\bar u_{k}(x)\,\hat x\!\cdot\!\nabla u_{k}(x)d\sigma(x)\right)=0\,.
$$
Thus, by \cite[Lemma 9.9]{McLe}, if $u_{k}$ satisfies the Sommerfeld radiation condition then $u_{k}(x)=0$ for any $\|x\|>R$ . Since $R$ is arbitrary, this gives $u_{k}=0$, contradicting our assumption $u\not=0$.
\end{proof}
By Theorem \ref{eigenfunctions} and by considering the usual family 
of generalized eigenfunctions  
$u^{\circ}_{\xi}\in H^{2}_{-\alpha}(\RE^{n})$, $\alpha>\frac{n}2$, 
of $\Delta:H^{2}(\RE^{n})\subset L^{2}(\RE^{n})\to L^{2}(\RE^{n})$ given by the plane waves
$$
u^{\circ}_{\xi}(x):=e^{ i\,\xi\cdot x}\,,
$$ 
one obtains the two families of generalized eigenfunctions of $\Delta_{\Pi,\Theta}$ defined by 
$$
u_{\xi}^{\pm}:= u^{\circ}_{\xi}+G_{-k^{2}}^\mp\Pi^{\prime}(  \Theta+\Pi M_{-k^{2}}^{\mp}\Pi^{\prime})  ^{-1}\Pi\gamma u^{\circ}_{\xi}\,,\qquad  k=\|\xi\|\,,\ 
-k^{2}\in E^{-}_{\Pi,\Theta}\,.
$$
\begin{remark}
Since $\ran(R^{\pm}_{-k^{2}})\subseteq H^{2}_{-\alpha}(\RE^{n})$, by \eqref{G_z_jumps} and \eqref{G2}, one has 
$$
[\gamma]G^{\pm}_{-k^{2}}=\uno_{H^{-{\frac32}}(\Gamma)\oplus H^{-{\frac12}}(\Gamma)}\,.
$$ 
Thus, since $[\gamma]u^{\circ}_{\xi}=0$, one gets
\be\label{LS0}
[\gamma]u_{\xi}^{\pm}=(  \Theta+\Pi M_{-k^{2}}^{\mp}\Pi^{\prime})  ^{-1}\Pi\gamma u^{\circ}_{\xi}\,,
\ee
and so the functions $u^{\pm}_{\xi}\in\dom(\t\Delta_{\Pi,\Theta})$ solve the Lippmann-Schwinger type equation
\be\label{LS}
u_{\xi}^{\pm}=u^{\circ}_{\xi}+G_{-k^{2}}^\mp
[\gamma]u_{\xi}^{\pm}\,.
\ee
\end{remark}
Let us now define, for any $u\in L^{2}_{\alpha}(\RE^{n})$,
\begin{align*}
&\F^{\circ}_{\pm}u(\xi)
:=\frac1{(2\pi)^{\frac{n}2}}\,\int_{\RE^{n}}\overline{u^{\pm}_{\xi}}(x)\,u(x)\,dx\\
=&Fu(\xi)+\frac1{(2\pi)^{\frac{n}2}}\,\langle(  \Theta+\Pi M_{-k^{2}}^{\mp}\Pi^{\prime})  ^{-1}\Pi\gamma u^{\circ}_{\xi}),\gamma R^{\pm}_{-k^{2}}u\rangle\\
=&
Fu(\xi)+\frac1{(2\pi)^{\frac{n}2}}\,\langle [\gamma]u_{\xi}^{\pm} ,\gamma R^{\pm}_{-k^{2}}u\rangle\,,
\end{align*}
where $ F$ denotes the Fourier transform and  $\langle\cdot,\cdot\rangle$ denotes the $(H^{-s_{1}}(\Gamma)\oplus H^{-s_{2}}(\Gamma))$-$(H^{s_{1}}(\Gamma)\oplus H^{s_{2}}(\Gamma))$ duality.
Next theorem provides the main properties of the maps $F^{\circ}_{\pm}$:
\begin{theorem}\label{scattering} 
1) The $\F^{\circ}_{\pm}$ extend to bounded operators $F_{\pm}\in\B(L^{2}(\RE^{n}))$ such that 
$\ker(\F_{\pm})=L^{2}(\RE^{n})_{pp}$ and $\F_{\pm}|L^{2}(\RE^{n})_{ac}$ are unitary onto $L^{2}(\RE^{n})$. 
\par\noindent
2) Let $P_{ac}$ be the orthogonal projection onto $L^{2}(\RE^{n})_{ac}$, then 
$$
\forall u\in \dom(\Delta_{\Pi,\Theta})\,,\quad (F_{\pm}P_{ac}\,\Delta_{\Pi,\Theta}u)(\xi)=-\|\xi\|^{2}F_{\pm} u(\xi)\,.
$$
3) Assume either \eqref{Theta_reg} or \eqref{Theta_reg2} holds, so that the wave operators 
$$
W_{\pm}
:=\text{s-}\lim_{t\to\pm\infty}e^{-it\Delta_{\Pi,\Theta}}e^{it\Delta}
$$
exist and are complete; then 
\be\label{Omega}
W_{\pm}=F_{\pm}^{*}F\,.
\ee
\end{theorem}
\begin{proof} 
1) We adapt to our framework the  reasonings in \cite[Section 6]{Agm} (see also \cite{Ike}). By \eqref{Krein}, one has, for any $z\in Z_{\Pi,\Theta}$,
\begin{align*}
&( F(-\Delta_{\Pi,\Theta}+z)^{-1}u)(\xi)\\
=&(\,\|\xi\|^{2}+z)^{-1}Fu(\xi)+
(FG_{z}\Pi^{\prime}(\Theta+\Pi M_{z}\Pi^{\prime})^{-1}\Pi\gamma(-\Delta+z)^{-1}u)(\xi)\\
=&(\,\|\xi\|^{2}+z)^{-1}Fu(\xi)\\
&+\frac{1}{(2\pi)^{n/2}}\,
\langle(\Theta+\Pi M_{\bar z}\Pi^{\prime})^{-1}\gamma(-\Delta+\bar z)^{-1}u^{\circ}_{\xi},\gamma(-\Delta+z)^{-1}u)\rangle\\
=&(\,\|\xi\|^{2}+z)^{-1}\breve u_{z}(\xi)\,,
\end{align*}
where
$$
\breve u_{z}(\xi):=Fu(\xi)+\frac{1}{(2\pi)^{n/2}}\,
\langle(\Theta+\Pi M_{\bar z}\Pi^{\prime})^{-1}\gamma u^{\circ}_{\xi},\gamma(-\Delta+z)^{-1}u)\rangle
\,.	
$$
Then, for $u\in L^{2}_{\alpha}(\RE^{n})$, we set
$$
\breve u^{\pm}_{-k^{2}}(\xi):=\lim_{\epsilon\downarrow 0}\breve u_{-(k^{2}\pm i\epsilon)}(\xi)\,.
$$
By the Theorems \ref{LAP0} and \ref{teo_LAP}, such a definition is well-posed and
$$\breve u^{\pm}_{-k^{2}}(\xi)=
Fu(\xi)+\frac{1}{(2\pi)^{n/2}}\,
\langle(\Theta+\Pi M^{\mp}_{-k^{2}}\Pi^{\prime})^{-1}\gamma u^{\circ}_{\xi},\gamma R^{\pm}_{-k^{2}}u\rangle\,,
$$
so that, for any $u\in L^{2}_{\alpha}(\RE^{n})$, 
$$
\breve u^{\pm}_{-\|\xi\|^{2}}(\xi)=\F^{\circ}_{\pm}u(\xi)\,.
$$
Then, one has (see \cite[page 191]{Agm} for the reasonings that allow the exchange $\lim\int=\int\lim$)
\begin{align*}
&\langle (E_{b}-E_{a})u,u\rangle_{L^{2}(\RE^{n})}\\
=&\lim_{\epsilon\downarrow 0}\frac{\epsilon}{\pi }\int_{a}^{b}
\|(-\Delta_{\Pi,\Theta}+(\lambda\pm i\epsilon))^{-1}u\|_{L^{2}(\RE^{n})}^{2}\,d\lambda
\\
=&\lim_{\epsilon\downarrow 0}\frac{\epsilon}{\pi }\int_{a}^{b}
\|F(-\Delta_{\Pi,\Theta}+\lambda\pm i\epsilon)^{-1}u)\|_{L^{2}(\RE^{n})}^{2}\,d\lambda\\
=&\lim_{\epsilon\downarrow 0}\frac{\epsilon}{\pi }\int_{a}^{b}\left(\int_{\RE^{n}}|\,\|\xi\|^{2}+\lambda\pm i\epsilon|^{-2}|\breve u_{\lambda\pm i\epsilon}(\xi)|^{2}\,d\xi\right)d\lambda\\
=&\int_{\RE^{n}}\left(\lim_{\epsilon\downarrow 0}\frac{\epsilon}{\pi }\int_{a}^{b}|\,\|\xi\|^{2}+\lambda\pm i\epsilon|^{-2}|\breve u_{\lambda\pm i\epsilon}(\xi)|^{2}\,d\lambda\right)d\xi\,.
\end{align*}
By the known properties of the Poisson integral (see e.g. \cite[equation (6.16)]{Agm}),
\be\begin{split}
\lim_{\epsilon\downarrow 0}&\frac{\epsilon}{\pi }\int_{a}^{b}|\,\|\xi\|^{2}+\lambda\pm i\epsilon|^{-2}|\breve u_{\lambda\pm i\epsilon}(\xi)|^{2}\,d\lambda\\
&=\begin{cases}
|\breve u_{-\|\xi\|^{2}}(\xi)|^{2}\,,& a<-\|\xi\|^{2}<b\\
0\,,&\text{otherwise}\,.\end{cases}
\end{split}
\ee
Therefore 
\be\label{Eba}\begin{split}
\langle (E_{b}-E_{a})u,u\rangle_{L^{2}(\RE^{n})}=&\int_{a<-\|\xi\|^{2}<b}|\breve u_{-\|\xi\|^{2}}(\xi)|^{2}\,d\xi\\=&\int_{a<-\|\xi\|^{2}<b}|\F^{\circ}_{\pm}u(\xi)|^{2}\,d\xi
\end{split}\ee
and so, if $P_{ac}$ denotes the orthogonal projector onto $L^{2}(\RE^{n})_{ac}$,  for any $u\in L^{2}_{\alpha}(\RE^{n})$ one has
\be\label{iso}
\|P_{ac}\,u\|^{2}_{L^{2}(\RE^{n})}=\langle P_{ac}\,u,u\rangle_{L^{2}(\RE^{n})}=\int_{\RE^{n}}|\F^{\circ}_{\pm}u(\xi)|^{2}\,d\xi=\|\F^{\circ}_{\pm}u\|^{2}_{L^{2}(\RE^{n})}\,.
\ee
This shows that $\F^{\circ}_{\pm}$ can be extended by continuity to a bounded map $F_{\pm}\in\B(L^{2}(\RE^{n})$. By \eqref{iso}, one gets $\ker(\F_{\pm})=L^{2}(\RE^{n})_{pp}$ and $\F_{\pm}$ is an isometry from $L^{2}(\RE^{n})_{ac}$ into $L^{2}(\RE^{n})$.  By Theorem \ref{TeoSpec}, $\ran(W_{\pm}^{*})=L^{2}(\RE^{n})$ and so $\ran(F_{\pm})=L^{2}(\RE^{n})$ will be a consequence of \eqref{Omega} which will be proven below. 
\par 
2) By \eqref{Eba} and by the polarization identity, for any $u,v\in L^{2}(\RE^{n})_{ac}$,
\begin{align*}
\langle E_{\lambda}u,v\rangle_{L^{2}(\RE^{n})}=\int_{-\|\xi\|^{2}<\lambda}\overline{F_{\pm}u}(\xi)F_{\pm}v(\xi)\,d\xi
\end{align*}
and so, for any $u,v\in L^{2}(\RE^{n})_{ac}\cap\dom(\Delta_{\Pi,\Theta})$,
\begin{align*}
&\langle \Delta_{\Pi,\Theta}u,v\rangle_{L^{2}(\RE^{n})}=\int_{-\infty}^{0}\lambda
\langle E_{\lambda}u,v\rangle_{L^{2}(\RE^{n})}\,d\lambda\\
=&-
\int_{\RE^{n}}\|\xi\|^{2}\overline{F_{\pm}u}(\xi)F_{\pm}v(\xi)\,d\xi=
-\langle\, \|\cdot\|^{2}F_{\pm}u,F_{\pm}v\rangle_{L^{2}(\RE^{n})}\,.
\end{align*}
\par 3) We equivalently show that $F_{\pm}W_{\pm}u= Fu$ for any $u$ in the Schwartz space of rapidly decreasing functions. Let define  $W_{\pm}(t):=P_{ac}\,e^{-it\Delta_{\Pi,\Theta}}e^{it\Delta}$. Since we are assuming the existence of the strong limits which define $W_{\pm}$, such limits can be replaced by the Abelian ones (see e.g. Corollary 14 and Lemma 15 in \cite[Section 6.1.2]{BW}); therefore
\begin{align*}
&F_{\pm}W_{\pm}u=
\lim_{\epsilon\to 0\pm} \epsilon\int_{0}^{\pm\infty}e^{-\epsilon t}
F_{\pm}W_{\pm}(t)u\,dt\\
=&
\lim_{\epsilon\to 0\pm} \int_{0}^{\pm\infty}e^{-\epsilon t}\,
\frac{d\ }{dt}F_{\pm}W_{\pm}(t)u\,dt+F_{\pm}u\,.
\end{align*}
The map $F_{\pm}$ diagonalize $P_{ac}\,\Delta_{\Pi,\Theta}$, thus
$$
(F_{\pm}W_{\pm}(t) u)(\xi)=\frac1{(2\pi)^{n/2}}\,\big\langle {u^{\pm}_{\xi}},e^{it(\Delta+\|\xi\|^{2})}u\big\rangle
$$ 
(here and below $\langle\cdot,\cdot\rangle$ denotes the $L^{2}_{-\alpha}$-$L^{2}_{\alpha}$ duality).  Since $(\Delta+\|\xi\|^{2}){u^{\circ}_{\xi}}=0$ and $u^{\pm}_{\xi}=u^{\circ}_{\xi}+G^{\mp}_{-\|\xi\|^{2}}[\gamma]u^{\pm}_{\xi}$, we get
\begin{align*}
&e^{-\epsilon t}\,\frac{d\ }{dt}\,(F_{\pm}W_{\pm}(t) u)(\xi)
=\frac{ie^{-\epsilon t}}{(2\pi)^{n/2}}\,\big\langle{u^{\pm}_{\xi}},(\Delta+\|\xi\|^{2})e^{it(\Delta+\|\xi\|^{2})}u\big\rangle\\
=&\frac{ie^{-\epsilon t}}{(2\pi)^{n/2}}\left(\big\langle(\Delta+\|\xi\|^{2}){u^{\circ}_{\xi}},e^{it(\Delta+\|\xi\|^{2})}u\big\rangle\right.\\
&\left.+\big\langle G^{\mp}_{-\|\xi\|^{2}}[\gamma]u^{\pm}_{\xi},(\Delta+\|\xi\|^{2})e^{it(\Delta+\|\xi\|^{2})}u\big\rangle\right)\\
=&\frac{i}{(2\pi)^{n/2}}\,\big\langle G^{\mp}_{-\|\xi\|^{2}}[\gamma] u^{\pm}_{\xi}, 
(\Delta+\|\xi\|^{2})e^{it(\Delta+\|\xi\|^{2}+i\epsilon)}u\big\rangle
\\
=&\frac{1}{(2\pi)^{n/2}}\,\frac{d\ }{dt}\,\big\langle G^{\mp}_{-\|\xi\|^{2}}[\gamma] u^{\pm}_{\xi}, 
(\Delta+\|\xi\|^{2})(\Delta+\|\xi\|^{2}+i\epsilon)^{-1}e^{it(\Delta+\|\xi\|^{2}+i\epsilon)}u\big\rangle.
\end{align*}
Therefore, a.e. (eventually taking the limit along a subsequence)
\begin{align*}
&(F_{\pm}W_{\pm}u)(\xi)\\=&\lim_{\epsilon\to 0\pm}\int_{0}^{\pm\infty}e^{-\epsilon t}\,\frac{d\ }{dt}\,(F_{\pm}W_{\pm}(t) u)(\xi)\,dt+ F_{\pm}u(\xi)\\
=&-\frac{1}{(2\pi)^{n/2}}\lim_{\epsilon\to 0\pm}\big\langle G^{\mp}_{-\|\xi\|^{2}}[\gamma] u^{\pm}_{\xi}, 
(\Delta+\|\xi\|^{2})(\Delta+\|\xi\|^{2}+i\epsilon)^{-1}u\big\rangle+ F_{\pm}u(\xi)\\
=&-\frac{1}{(2\pi)^{n/2}}\,\big\langle G^{\mp}_{-\|\xi\|^{2}}[\gamma] u^{\pm}_{\xi}, 
u\big\rangle+ F_{\pm}u(\xi)\\
=&Fu(\xi)\,.
\end{align*}
\end{proof}
Let us now introduce the scattering operator $S:=W_{+}^{*}W_{-}$,  so that, by $W_{\pm}=F^{*}_{\pm}F$, one gets $FSF^{*}=F_{+}F_{-}^{*}$. The scattering matrix 
$$
S_k:L^{2}(\Sf^{n-1})\to L^{2}(\Sf^{n-1})\,,
$$
where $\Sf^{n-1}$ denotes the unit sphere in $\RE^{n}$, 
is then defined  by the relation
$$
S_{k}(Fu)_{k}=(F{Su})_{k}\,,\quad (Fu)_{k}(\hat\xi):=F u(k\hat\xi)\,.
$$
Therefore
$$
S_{k}(\F_{\!-}u)_{k}=(\F_{\!+}u)_{k}\,.
$$
The next results shows how the kernel (proportional to the scattering amplitude) of the linear operator $\uno-S_{k}$ can be expressed  in terms of the 
limit Weyl functions $\Theta+\Pi M^{\pm}_{-k^{2}}\Pi'$; here $\mu$ denotes Lebesgue measure on 
$\Sf^{n-1}$.
\begin{theorem}\label{SM} Assume either \eqref{Theta_reg} or \eqref{Theta_reg2} holds, so that $W_\pm$ and hence $S_{k}$ exist. Then, for any $k>0$ such that $-k^{2}\in E^{-}_{\Pi,\Theta}$,
\begin{align*}
S_{k}f(\hat\xi)=f(\hat\xi)-\int_{\Sf^{n-1}}s_{k}(\hat\xi,\hat \xi')\, f(\hat \xi')\,d\mu(\hat \xi')\,,
\end{align*}
where
\begin{align*}
s_{k}(\hat\xi,\hat \xi'):=&\frac{i}{4\pi}\,\left(\frac{k}{2\pi}\right)^{n-2}\langle \Pi\gamma u^{\circ}_{k\hat\xi'},(  \Theta+\Pi M_{-k^{2}}^{-}\Pi^{\prime})  ^{-1}\Pi\gamma u^{\circ}_{k\hat\xi}\rangle\\
=&\frac{i}{4\pi}\,\left(\frac{k}{2\pi}\right)^{n-2}\langle (  \Theta+\Pi M_{-k^{2}}^{+}\Pi^{\prime})  ^{-1}\Pi\gamma u^{\circ}_{k\hat \xi'},\Pi\gamma u^{\circ}_{k\hat\xi}\rangle\,.
\end{align*}
In the case \eqref{Theta_reg} holds, one also has the equivalent representation
\begin{align*}
s_{k}(\hat\xi,\hat \xi')=&\frac{i}{4\pi}\,\left(\frac{k}{2\pi}\right)^{n-2}\langle \Pi\gamma u^{\circ}_{k\hat\xi'},(B_\Theta-\Pi \gamma G^{-}_{-k^{2}}\Pi^{\prime})^{-1}\Pi\gamma u^{\circ}_{k\hat\xi}\rangle\\
=&\frac{i}{4\pi}\,\left(\frac{k}{2\pi}\right)^{n-2}\langle (B_\Theta-\Pi \gamma G^{+}_{-k^{2}}\Pi^{\prime})^{-1}\Pi\gamma u^{\circ}_{k\hat\xi'},\Pi\gamma u^{\circ}_{k\hat\xi}\rangle
\,.
\end{align*}
\end{theorem}
\begin{proof} Here we follow the same strategy as in \cite{Sch} and \cite{AS}.
By the definition of $S_{k}$ we only need to show that, for any $u\in L^{2}_{\alpha}(\RE^{n})$, one has
\be\label{eq1}
(\F_{+}u)_{k}(\hat\xi)=(\F_{-}u)_{k}(\hat\xi)-\int_{\Sf^{n-1}}
 s_{k}(\hat\xi,\hat\xi')\, (\F_{-}u)_{k}(\hat\xi')
\,d\mu(\hat\xi')\,.
\ee
Let us define the auxiliary functions
$$
v_{k\hat\xi}:=u^{-}_{k\hat\xi}-u^{+}_{k\hat\xi}-\int_{\Sf^{n-1}}
 s_{k}(\hat\xi,\hat\xi')\, 
u^{-}_{k\hat\xi'}\,d\mu(\hat\xi')\,.
$$
By \eqref{LS}, 
\begin{align*}
v_{k\hat\xi}=&G_{-k^{2}}^{+}[\gamma]u^{-}_{k\hat\xi}-G_{-k^{2}}^{-}[\gamma]u^{+}_{k\hat\xi}
-\int_{\Sf^{n-1}}s_{k}(\hat\xi,\hat\xi')\,G_{-k^{2}}^{+}[\gamma]u^{-}_{k\hat\xi'}
 \, d\mu(\hat\xi')\\
-&\frac{i}{4\pi}\,\left(\frac{k}{2\pi}\right)^{n-2}\int_{\Sf^{n-1}}\langle \gamma u^{\circ}_{k\hat\xi'},[\gamma]u^{+}_{k\hat\xi}\rangle u^{\circ}_{k\hat\xi'}\, d\mu(\hat\xi')\,.
\end{align*}
By
$$
\int_{\Sf^{n-1}} \bar u^{\circ}_{k\hat\xi}(x)\,u^{\circ}_{k\hat\xi}(y)\,d\mu(\hat\xi)=4\pi i\left(\frac{2\pi}{k}\right)^{n-2}\!\!\left(\G^{-}_{-k^{2}}(x-y)
-\G^{+}_{-k^{2}}(x-y)\right)$$
(see \cite[formula (15)]{AS}) and by \eqref{K1}-\eqref{K2}, one gets
$$
-\frac{i}{4\pi}\,\left(\frac{k}{2\pi}\right)^{n-2}\!\!\int_{\Sf^{n-1}}\langle \gamma u^{\circ}_{k\hat\xi'},[\gamma]u^{+}_{k\hat\xi}\rangle\,u^{\circ}_{k\hat\xi'}
 \,d\mu(\hat\xi')=\left(G^{-}_{-k^{2}}-G^{+}_{-k^{2}}\right)
 [\gamma]u^{+}_{k\hat\xi}
$$
and so 
$$
v_{k\hat\xi}=G^{+}_{-k^{2}}[\gamma]v_{k\hat\xi}\,.
$$
Therefore, by 1) in Lemma \ref{SRC}, $v_{k\hat\xi}$ satisfies the Sommerfeld radiation condition. 
Since $u^{\pm}_{k\hat\xi}\in\ker(\t\Delta_{\Pi,\Theta}+k^{2})$, one has $v_{k\hat\xi}\in\ker(\t\Delta_{\Pi,\Theta}+k^{2})$. Thus, 
 by 2) in Lemma \ref{SRC}, $v_{k\hat\xi}=0$ and so 
$$
u^{+}_{k\hat\xi}=u^{-}_{k\hat\xi}-\int_{\Sf^{n-1}}
 s_{k}(\hat\xi,\hat\xi')\, 
u^{-}_{k\hat\xi'}\,d\mu(\hat\xi')\,.
$$
Considering the duality product of both the left and right functions with $u\in L^{2}_{\alpha}(\RE^{n})$, one gets \eqref{eq1} and the proof is done.
\end{proof}
\begin{remark}\label{BY} Given $\mu\in (0,+\infty)\cap\rho(\Delta_{\Pi,\Theta})$, let $W^{\mu}_{\pm}$ denote the wave operators for the scattering couple $\big((-\Delta_{\Pi,\Theta}+\mu)^{-1}, (-\Delta+\mu)^{-1}\big)$. Since both $W_{\pm}$ and $W^{\mu}_{\pm}$ exist and are complete, by the Birman-Kato invariance principle one gets $W^{\mu}_{\pm}=W_{\pm}$. By \eqref{Krein} one has 
\be\label{RD}
(-\Delta_{\Pi,\Theta}+\mu)^{-1}-(-\Delta+\mu)^{-1}=G_{\mu}\Pi'(\Theta+\Pi M_{\mu}\Pi')^{-1}\Pi G_{\mu}'\,.
\ee
The Birman-Yafaev general scheme in stationary scattering theory (see e.g. \cite{BY}, \cite{Y}, \cite{Y1}), conditional on the existence of the limit operator $$
B^{+}_{\lambda}:=\lim_{\epsilon\downarrow 0}B_{\lambda+i\epsilon}\,,\quad
B_{z}:=\Pi G_{\mu}'\big((-\Delta+\mu)^{-1}-z\big)^{-1}G_{\mu}\Pi'
$$
and of the inverse $(\uno+B^{+}_{\lambda}(\Theta+\Pi M_{\mu}\Pi')^{-1})^{-1}$,
allows the representation formula for the scattering matrix $S^{\mu}_{\lambda}$, corresponding to the scattering operator $S^{\mu}= (W_{+}^{\mu})^{*}W^{\mu}_{-}$, given by (see e.g. \cite[equation (2.8)]{Y1}) 
\begin{equation}\label{S}
S^{\mu}_{\lambda}=\uno-2\pi iL_{\lambda}\Pi'(\Theta+\Pi M_{\mu}\Pi')^{-1}(\uno+B^{+}_{\lambda}(\Theta+\Pi M_{\mu}\Pi')^{-1})^{-1}\Pi L_{\lambda}'\,.
\end{equation}
Here  $$
L_\lambda: H^{-3/2}(\Gamma)\oplus H^{-1/2}(\Gamma)\to L^{2}(\Sf^{n-1})\,,$$
$$ (L_{\lambda}(\phi\oplus\varphi))(\hat\xi):=[(F_{0}G_{\mu}(\phi\oplus\varphi))(\lambda)](\hat\xi)$$ 
is defined in terms of the unitary map $F_{0}:L^{2}(\RE^{n})\to L^{2}((0,\mu^{-1});\Sf^{n-1})$ such that the operator $F_{0}(-\Delta+\mu)^{-1}F^{*}_{0}$ acts as multiplication by $\lambda$, i.e. 
$$
[(F_{0}u)(\lambda)](\hat\xi):=2^{-\frac12}\left(\frac1\lambda-\mu\right)^{\frac{n-2}4}(Fu)\left(\left(\frac1\lambda-\mu\right)^{\frac12}\hat\xi\right)\,.
$$
By the identities 
$$
\left((-\Delta+\mu)^{-1}-z\right)^{-1}=-\frac1z\,\left(\uno+\frac1{z}\,\left(-\Delta+\mu-\frac1z\right)^{-1}\right)\,,
$$
$$
G'_{\mu}\left(-\Delta+\mu-\frac1z\right)^{-1}=z\left(G_{\mu-\frac1z}'-G_{\mu}'\right)\,,
$$
$$
M_{\mu-\frac1z}=M_{\mu}-\frac1z\,G_{\mu-\frac1z}'G_{\mu}\,,
$$
one obtains
\begin{align*}
&(\Theta+\Pi M_{\mu}\Pi')^{-1}(\uno+B_{z}(\Theta+\Pi M_{\mu}\Pi')^{-1})^{-1}=
(\Theta+\Pi M_{\mu}\Pi'-B_{z})^{-1}\\
=&\left(\Theta+\Pi M_{\mu}\Pi'-\Pi G_{\mu}'\left(\frac{\uno}z+\frac1{z^{2}}\,\left(-\Delta+\mu-\frac1z\right)^{-1}\right)G_{\mu}\Pi'\right)^{-1}\\
=&\left(\Theta+\Pi M_{\mu}\Pi'-\frac1z\,\Pi G_{\mu}'G_{\mu}\Pi'-\frac1z\,\Pi\left(G_{\mu-\frac1z}'-G_{\mu}'\right)G_{\mu}\Pi'\right)^{-1}\\
=&\left(\Theta+\Pi M_{\mu}\Pi'-\frac1z\,\Pi G_{\mu-\frac1z}'G_{\mu}\Pi'\right)^{-1}=\left(\Theta+\Pi M_{\mu-\frac1z}\Pi'\right)^{-1}\,.
\end{align*}
Therefore, by Theorems \ref{LAP} and \ref{teo_LAP}, both $B^{+}_{\lambda}$ and  $(\uno+B^{+}_{\lambda}(\Theta+\Pi M_{\mu}\Pi')^{-1})^{-1}$ are well defined and 
\be\label{S1}
S^{\mu}_{\lambda}=\uno-2\pi iL_{\lambda}\Pi'\big(\Theta+\Pi M^{+}_{\mu-\frac1\lambda}\Pi'\big)^{-1}\Pi L_{\lambda}'\,,\qquad \mu-\frac1\lambda\in E^{-}_{\Pi,\Theta}\,.
\ee
In case Theorems \ref{LAP} and \ref{teo_LAP} were not available, using the results contained in \cite[Chapter 7, Sections 4 and 6]{Y}, the representation formula \eqref{S} could be still obtained under Kato-smoothness or trace-class hypotheses on the resolvent difference \eqref{RD}. However for the models we are here considering, the trace-class condition is not always fulfilled while checking the smoothness property may be a substantial problem (see e.g. \cite[Chapter 7, Sections 4, Proposition 1]{Y}).
\par Finally, using the correspondence $S^{\mu}_{\lambda}=S_{k}$, which holds whenever $\mu-\frac1\lambda=-k^{2}$ (see \cite[Section 6, Chapter 2]{Y}), and the identity
$$
(L_{\lambda}(\phi\oplus\varphi))(\hat\xi)=
2^{-\frac12}\left(\frac1\lambda-\mu\right)^{\frac{n-2}4}{(2\pi)^{-\frac{n}2}}\,\left\langle \gamma u^{\circ}_{\left(\lambda^{-1}-\mu\right)^{\frac12}\hat\xi},\phi\oplus\varphi\right\rangle\,,
$$
one gets that \eqref{S1} matches the formula provided in Theorem \ref{SM}. 
\end{remark}
\end{section}

\begin{section}{Examples: Boundary conditions on $\Gamma$}
In this section we apply our results to self-adjoint realizations of the Laplacian with various kind of boundary conditions on $\Gamma$. For more details on such models we refer to \cite[Section 5]{MaPoSi}. In particular, by the results given there, hypothesis \eqref{Theta_reg} holds for all the examples presented  here. As regards the semi-boundedness hypotheses required in Theorem \ref{Theorem_LAP}, the semi-boundedness of the operators $\Delta_{D}$ and $\Delta_{N}$ in subsections \ref{dirichlet} and \ref{neumann} is clear, semi-boundedness of  $\Delta_{R}$ in subsection \ref{robin} is provided in \cite[Remark 5.2]{MaPoSi} and semiboundedness of $\Delta_{\alpha,\delta}$ and $\Delta_{\beta,\delta'}$ in subsections \ref{delta} and \ref{deltaprimo} is provided in  \cite[Theorem 3.16]{BLL} (see also the next Section, the proofs being essentially the same). In the following, in order to simplify the exposition, we suppose that $\Omega_{\+}$ is connected.
\begin{subsection}{Dirichlet boundary conditions.}\label{dirichlet} 
Let us consider the self-adjoint extension $\Delta_{D}$ corresponding to Dirichlet boundary conditions on the whole $\Gamma$; it is given by the direct sum $\Delta_{D}=\Delta^{D}_{{\-}}\oplus \Delta^{D}_{{\+}}$, where the self-adjoint operators $\Delta^{D}_{\-}$ and $\Delta^{D}_{\+}$ are defined by $\Delta^{D}_{\-}:=\Delta|\dom(\Delta^{D}_{\-})$ and $\Delta^{D}_{\+}:=\Delta|\dom(\Delta^{D}_{\+})$, with  domains
$\dom(\Delta^{D}_{\-})=\{u_{\-}\in H^{2}(\Omega_{\-}):\gamma^{\-}_{0}u_{\-}=0\}$  
and  
$\dom(\Delta^{D}_{\+})=\{u_{\+}\in H^{2}(\Omega_{\+}):\gamma^{\+}_{0}u_{\+}=0\}$. 
Since 
\begin{align*}
\dom(\Delta^{D}_{\-})\oplus\dom(\Delta^{D}_\+)=&\{u\in H^{2}(\RE^{n}\backslash\Gamma): [\gamma_{0}]u=0\,,\gamma_{0}u=0\}\\
=&\{u\in H^{1}(\RE^{n})\cap H^{2}(\RE^{n}\backslash\Gamma):\gamma_{0}u=0\}\,,
\end{align*}
that corresponds, in Corollary \ref{Lemma_parameter}, to the choice  $\Pi(\phi\oplus\varphi):=\phi\oplus 0$, and 
$B_{\Theta}=0$.
Thus (see \cite[Subsection 5.1]{MaPoSi})
$$
(\Delta^{D}_{{\-}}\oplus \Delta^{D}_{{\+}})u=\Delta u-[\gamma_{1}]u\,\delta_{\Gamma}
$$
and, by $(\gamma_{0}\SL_{z})^{-1}=P_{z}^{\-}-P_{z}^{\+}$, where $P_{z}^{\-}$ and $P_{z}^{\+}$ denote the Dirichlet-to-Neumann operators for $\Omega_{\-}$ and $\Omega_{\+}$ respectively (see e.g. \cite[equation (5.4)]{MaPoSi}), one has, for any $z\in \CO\backslash(-\infty,0]$,
\begin{align*}
(-(\Delta^{D}_{{\-}}\oplus \Delta^{D}_{{\+}})+z)^{-1}=&
(-\Delta+z)^{-1}+\SL_{z}(P_{z}^{\+}-P^{\-}_{z})\gamma_{0}(-\Delta+z)^{-1}\,.
\end{align*}
Then, by Theorem \ref{SM}, one has, for any $k>0$ such that $-k^{2}\notin \sigma(\Delta^{D}_{\-})$,
\begin{align*}
s_{k}(\hat\xi,\hat \xi')
=&\frac{i}{4\pi}\,\left(\frac{k}{2\pi}\right)^{n-2}\langle (P_{-k^{2}}^{\+}-P^{\-}_{-k^{2}})^{+}\gamma_{0} u^{\circ}_{k\hat\xi'},\gamma_{0} u^{\circ}_{k\hat\xi}\rangle\,,
\end{align*}
where 
$$
(P_{-k^{2}}^{\+}-P^{\-}_{-k^{2}})^{+}:=\lim_{\epsilon\downarrow 0}\, (P_{-k^{2}+i\epsilon}^{\+}-P^{\-}_{-k^{2}+i\epsilon})=-(\gamma_{0}\SL^{+}_{-k^{2}})^{-1}\,.
$$
Such a limit exists in $\B(H^{{\frac32}}(\Gamma), H^{-{\frac32}}(\Gamma))$ by Theorem \ref{teo_LAP}. Notice that, restricted to the case $n=2$, similar formulae have been obtained (without the smoothness condition on $\Gamma$) in \cite[Theorems 5.3 and 5.6]{EP1}.
\end{subsection}
\begin{subsection}{Neumann boundary conditions.}\label{neumann}  
Let us consider the self-adjoint extension $\Delta^{N}$ corresponding to Neumann boundary conditions on the whole $\Gamma$; it is given by the direct sum $\Delta_{N}=\Delta^{N}_{{\-}}\oplus \Delta^{N}_{{\+}}$, where 
the self-adjoint operators $\Delta^{N}_{\-}$ and $\Delta^{N}_{\+}$ are defined by $\Delta^{N}_{\-}:=\Delta|\dom(\Delta^{N}_{\-})$ and $\Delta^{N}_{\+}:=\Delta|\dom(\Delta^{N}_{\+})$, with domains 
$\dom(\Delta^{N}_{\-})=\{u_{\-}\in H^{2}(\Omega_{\-}):\gamma_{1}^{\-}u_{\-}=0\}$ and $\dom(\Delta^{N}_{\+})=\{u_{\+}\in H^{2}(\Omega_{\+}):\gamma_{1}^{\+}u_{\+}=0\}$. 
Since 
$$
\dom(\Delta^{N}_{\-})\oplus\dom(\Delta^{N}_\+)=\{u\in H^{2}(\RE^{n}\backslash\Gamma): [\gamma_{1}]u=\gamma_{1}u=0\}\,,
$$
that corresponds, in Corollary \ref{Lemma_parameter}, to the choice  $\Pi(\phi\oplus\varphi):=0\oplus\varphi$, and 
$B_{\Theta}=0$. Thus (see \cite[Subsection 5.2]{MaPoSi})
$$
(\Delta^{N}_{{\-}}\oplus \Delta^{N}_{{\+}})u=\Delta u-[\gamma_{0}]u\,\nu\!\cdot\!\nabla\delta_{\Gamma}\,,
$$
and, by $(\gamma_{1}\DL_{z})^{-1}=Q_{z}^{\+}-Q_{z}^{\-}$, where $Q_{z}^{\-}$ and $Q_{z}^{\+}$ denote the Neumann-to-Dirichlet operators for $\Omega_{\-}$ and $\Omega_{\+}$ respectively (see e.g. \cite[equation (5.7)]{MaPoSi}), one has, for any $z\in \CO\backslash(-\infty,0]$,
\begin{align*}
(-(\Delta^{N}_{{\-}}\oplus \Delta^{N}_{{\+}})+z)^{-1}=&
(-\Delta+z)^{-1}+\DL_{z}(Q_{z}^{\-}-Q^{\+}_{z})\gamma_{1}(-\Delta+z)^{-1}\,.
\end{align*}
Then, by Theorem \ref{SM}, one has, for any $k>0$ such that $-k^{2}\notin \sigma(\Delta^{N}_{\-})$,
\begin{align*}
s_{k}(\hat\xi,\hat \xi')
=&\frac{i}{4\pi}\,\left(\frac{k}{2\pi}\right)^{n-2}\langle (Q_{-k^{2}}^{\-}-Q^{\+}_{-k^{2}})^{+}\gamma_{1} u^{\circ}_{k\hat\xi'},\gamma_{1} u^{\circ}_{k\hat\xi}\rangle\,,
\end{align*}
where 
$$
(Q_{-k^{2}}^{\-}-Q^{\+}_{-k^{2}})^{+}:=\lim_{\epsilon\downarrow 0}\, (Q_{-k^{2}+i\epsilon}^{\-}-Q^{\+}_{-k^{2}+i\epsilon})=-(\gamma_{1}\DL^{+}_{-k^{2}})^{-1}\,.
$$
Such a limits exists in $\B(H^{{\frac12}}(\Gamma), H^{-{\frac12}}(\Gamma))$ by Theorem \ref{teo_LAP}. Notice that, restricted to the case $n=2$, similar formulae have been obtained (without the smoothness condition on $\Gamma$) in \cite[Theorems 4.2 and 4.3]{EP3}.
\end{subsection}  
\begin{subsection}{Robin boundary conditions.}\label{robin} Let us consider the self-adjoint extension $\Delta_{R}$ corresponding to Robin boundary conditions on the whole $\Gamma$; it is given by the direct sum $\Delta_{R}=\Delta^{R}_{{\-}}\oplus \Delta^{R}_{{\+}}$, where 
 $$\Delta^{R}_{\-}:=\Delta|\dom(A^{R}_{\-})\,,\qquad \Delta^{R}_{\+}:=\Delta|\dom(\Delta^{R}_{\+}) \,,
 $$
 $$
\dom(\Delta^{R}_{\-})=\{u_{\-}\in \dom(\Delta_{\-}^{\max}):\gamma_{1}^{\-}u_{\-}= b_{\-}\,\gamma_{0}^{\-}u_{\-}\}\,,
$$ 
$$
\dom(\Delta^{R}_{\+})=\{u_{\+}\in \dom(\Delta_{\+}^{\max}):\gamma_{1}^{\+}u_{\+}= b_{\+}\,\gamma_{0}^{\+}u_{\+}\}\,.
$$ 
Here $b_{\-}$ and $b_{\+}$ are real-valued multipliers in $H^{\frac12}(\Gamma)$. Since, in case $b_{\+}(x)\not=b_{\-}(x)$ for a.e. $x\in\Gamma$, the domain of $\Delta^{R}_{\-}\oplus \Delta^{R}_{\+}$ is given by 
\begin{align*}
&\dom(\Delta^{R}_{\-}\oplus \Delta^{R}_{\+})\\=
\big\{& u\in H^{2}(\RE^{n}\backslash\Gamma):(b_{\+}-b_{\-})\gamma_{0}u
=
[\gamma_{1}]u-\frac12(b_{\+}+b_{\-})[\gamma_{0}]u\,,\\
&
(b_{\+}-b_{\-}) \gamma_{1}u
=
\frac12(b_{\+}+b_{\-})[\gamma_{1}]u-b_{\+}b_{\-}[\gamma_{0}]u
\big\}\,.
\end{align*} 
that corresponds, in Corollary \ref{Lemma_parameter}, to the choice  $\Pi=1$ and $B_{\Theta}=B_{R}$, where
$$
B_{R}=-\frac{1}{[b]}\left[\,\begin{matrix}1  &\langle b\rangle  \\
\langle b\rangle  &{b_{\+}b_{\-}}
\end{matrix}\,\right]\,,
\quad
 \langle b\rangle:=\frac12(b_{\+}+b_{\-})\,,\quad \quad [b]:=b_{\+}-b_{\-}
\,,
$$
Thus  (see \cite[Subsection 5.3]{MaPoSi})
$$
(\Delta^{R}_{{\-}}\oplus \Delta^{R}_{{\+}})u=\Delta u-\frac4{[b]}\,\left((\langle b\rangle\,\gamma_{1}u-b_\+b_{\-}\gamma_{0}u)\,\delta_{\Gamma}+(\gamma_{1}u-\langle b\rangle\,\gamma_{0}u)\,\nu\!\cdot\!\nabla\delta_{\Gamma}\right)
$$
and, for any $z\in \rho(\Delta^{R}_{\-})\cap\rho(\Delta^{R}_{\+})\cap \CO\backslash(-\infty,0]$,
\begin{align*}
&(-(\Delta^{R}_{{\-}}\oplus \Delta^{R}_{{\+}})+z)^{-1}\\
=&
(-\Delta+z)^{-1}
-G_{z} \left[\,\begin{matrix}1/{[b]}+\gamma_{0}\SL_{z}&\langle b\rangle/[b]+\gamma_{ 0}\DL_{z}\\
\langle b\rangle/[b]+\gamma_{1}\SL_{z}&b_{+}b_{-}/[b]+\gamma_{1}\DL_{z}
\end{matrix}\,\right]^{-1}\gamma(-\Delta+z)^{-1}\,,
\end{align*}
where $G_{z}(\phi\oplus\varphi)=\SL_{z}\phi+\DL_{z}\varphi$. Let us notice that the case in which one has the same Robin boundary conditions on both sides of $\Gamma$ corresponds to the choice $b_{\+}=b_{\circ}=-b_{\-}$. 
\par
Then, by Theorem \ref{SM}, one has, for any $k>0$ such that $-k^{2}\notin \sigma(\Delta^{R}_{\-})$,
\begin{align*}
&s_{k}(\hat\xi,\hat \xi')\\
=&-\frac{i}{4\pi}\left(\frac{k}{2\pi}\right)^{n-2}\!\!\left\langle 
\left[\,\begin{matrix}1/{[b]}+\gamma_{0}\SL^{+}_{-k^{2}}&\langle b\rangle/[b]+\gamma_{ 0}\DL^{+}_{-k^{2}}\\
\langle b\rangle/[b]+\gamma_{1}\SL^{+}_{-k^{2}}&b_{+}b_{-}/[b]+\gamma_{1}\DL^{+}_{-k^{2}}
\end{matrix}\,\right]^{-1}\!\!\!\!\!\gamma u_{k\xi'}^{\circ},\gamma u_{k\xi}^{\circ}\right\rangle.
\end{align*}

\end{subsection}  
\begin{subsection}{$\delta$-interactions.}\label{delta} 
Here we consider the self-adjoint extension corresponding to the choice $\Pi(\phi\oplus\varphi)=\phi\oplus 0$ and 
$\Theta(\phi\oplus\varphi)=-(\phi/\alpha+\gamma_{0}\SL\phi)\oplus 0$, where  $\alpha$ is a real-valued multiplier in $H^{\frac32}(\Gamma)$ such that $1/\alpha\in L^{\infty}(\Gamma)$. Such a kind of self-adjoint extensions correspond to the boundary conditions ${\alpha}\gamma_{0}u=[\gamma_{1}]u$ and so one obtains  the self-adjoint extensions usually called ''$\delta$-interactions on $\Gamma\,$'' (see  \cite{BEK}, \cite{BLL} and references therein). By Corollary \ref{Lemma_parameter}  (see \cite[Subsection 5.4]{MaPoSi}), one gets the  self-adjoint extension
$$
\Delta_{\alpha,\delta}\,u=\Delta u-{\alpha}\gamma_{0}u\,\delta_{\Gamma}\,,
$$
$$
\dom(\Delta_{\alpha,\delta}):=\{u\in H^{1}(\RE^{n})\cap H^{2}(\RE^{n}\backslash\Gamma):{\alpha}\gamma_{0}u=[\gamma_{1}]u\}\,;
$$
its  resolvent is given by 
\begin{align*}
(-\Delta_{\alpha,\delta}+z)^{-1}=&
(-\Delta+z)^{-1}-\SL_{z}((1/\alpha)+\gamma_{0}\SL_{z})^{-1}\gamma_{0}(-\Delta+z)^{-1}
\\=&(-\Delta+z)^{-1}-\SL_{z}(1+\alpha\gamma_{0}\SL_{z})^{-1}\alpha\gamma_{0}(-\Delta+z)^{-1}
\,. 
\end{align*}
Then, by Theorem \ref{SM} and Remark \ref{negative2}, one has, for any $k>0$,
\begin{align*}
s_{k}(\hat\xi,\hat \xi')
=&-\frac{i}{4\pi}\,\left(\frac{k}{2\pi}\right)^{n-2}\langle (1+\alpha\gamma_{0}\SL^{+}_{-k^{2}})^{-1}\alpha\gamma_{0} u^{\circ}_{k\hat\xi'},\gamma_{0} u^{\circ}_{k\hat\xi}\rangle\,.
\end{align*}
\end{subsection}  

\begin{subsection}{$\delta'$-interactions.}\label{deltaprimo} 
Here we consider the self-adjoint extension corresponding to the choice $\Pi(\phi\oplus\varphi)=0\oplus\varphi$ and 
$\Theta(\phi\oplus\varphi)=0\oplus(\varphi/\beta-\gamma_{1}\DL\varphi)$, where  $\beta$ is a real-valued multiplier in $H^{\frac12}(\Gamma)$ such that $1/\beta\in L^{\infty}(\Gamma)$. Such a kind of self-adjoint extensions correspond to the boundary conditions ${\beta}\gamma_{1}u=[\gamma_{0}]u$ and so one obtains  the self-adjoint extensions usually called ''$\delta'$-interactions on $\Gamma\,$'' (see \cite{BLL} and references therein). 
By Corollary \ref{Lemma_parameter} (see \cite[Subsection 5.5]{MaPoSi}), one gets  the  self-adjoint extension
$$
\Delta_{\beta,\delta'}u=\Delta u-{\beta}\gamma_{1}u\,\nu\!\cdot\!\nabla\delta_{\Gamma}\,,
$$
$$
\dom(\Delta_{\beta,\delta'}):=\{u\in H^{2}(\RE^{n}\backslash\Gamma):[\gamma_{1}]u=0\,,\ {\beta}\gamma_{1}u=[\gamma_{0}]u\}\,;
$$
Its resolvent is given by
\begin{align*}
(-\Delta_{\beta,\delta'}+z)^{-1}=&
(-\Delta+z)^{-1}+\DL_{z}((1/\beta)-\gamma_{1}\DL_{z})^{-1}\gamma_{1}(-\Delta+z)^{-1}
\\=&(-\Delta+z)^{-1}+\DL_{z}(1-\beta\gamma_{1}\DL_{z})^{-1}\beta\gamma_{1}(-\Delta+z)^{-1}
\,.
\end{align*}
Then, by Theorem \ref{SM} and Remark \ref{negative2},, one has, for any $k>0$,
\begin{align*}
s_{k}(\hat\xi,\hat \xi')
=&\frac{i}{4\pi}\,\left(\frac{k}{2\pi}\right)^{n-2}\langle (1-\beta\gamma_{1}\DL^{+}_{-k^{2}})^{-1}\beta\gamma_{1} u^{\circ}_{k\hat\xi'},\gamma_{1} u^{\circ}_{k\hat\xi}\rangle\,.
\end{align*}

\end{subsection}  

\end{section}
\begin{section}{Examples: Boundary conditions on $\Sigma\subset\Gamma$}
In this section we consider  boundary conditions supported on a relatively open part $\Sigma\subset\Gamma$ with Lipschitz boundary. For more details and proof regarding such models we refer to \cite[Section 6]{MaPoSi}. In particular, by the results given there, hypothesis \eqref{Theta_reg2} holds for all the examples presented  here; moreover, the semi-boundedness hypothesis required in Theorem \ref{Theorem_LAP} holds true as well: this point is next discussed case-by-case. In order to apply Theorem \ref{negative1}, so to simplify the exposition, we suppose that $\RE^{n}\backslash\overline\Sigma$ is connected. \par 
In the following, given $X\subset\Gamma$ closed, we use the definition $$H^{s}_{X}(\Gamma):=\{\phi\in H^{s}(\Gamma):\supp(\phi)\subseteq X\}\,.$$ Given $\Sigma\subset \Gamma$ relatively open of class $\C^{0,1}$,  we denote by $\Pi_{\Sigma}$ the orthogonal projector in the Hilbert space $H^{s}(\Gamma)$, $s>0$, such that $\ran(\Pi_{\Sigma})=H^{s}_{\Sigma^{c}}(\Gamma)^{\perp}$. One has $\ran(\Pi_{\Sigma}')=H^{-s}_{\overline\Sigma}(\Gamma)$, where $\Pi'_{\Sigma}=\Lambda^{2s}\Pi_{\Sigma}\Lambda^{-2s}$ is the dual projection. In the following, we use the identifications $H^{s}_{\Sigma^{c}}(\Gamma)^{\perp}\simeq H^{s}(\Sigma)$ and $H^{-s}_{\overline\Sigma}(\Gamma)\simeq H^{s}(\Sigma)'$. In particular, by the former, the orthogonal projection $\Pi_{\Sigma}$ can be identified with the restriction map $R_{\Sigma}:H^{s}(\Gamma)\to H^{s}(\Sigma)$, $R_{\Sigma}\phi:=\phi|\Sigma$.
\begin{subsection}{Dirichlet boundary conditions.}\label{dir-loc} We denote by $\Delta_{D,\Sigma}$ the self-adjoint extension corresponding to the orthogonal projector defined by $\Pi(\phi\oplus\varphi):=(\Pi_{\Sigma}\phi)\oplus 0\equiv(\phi|\Sigma)\oplus 0$ and to the self-adjoint operator $\Theta(\phi\oplus\varphi):=(-\Theta_{D,\Sigma}\phi)\oplus 0$, 
$$
\Theta_{D,\Sigma}:\dom(\Theta_{D,\Sigma})\subseteq H^{-{\frac32}}_{\overline\Sigma}(\Gamma)\to 
H^{{\frac32}}(\Sigma)\,,\quad \Theta_{D,\Sigma}\phi:=(\gamma_{0}\SL\phi)|\Sigma\,,
$$
$$
\dom(\Theta_{D,\Sigma}):=\{\phi\in H^{-{\frac12}}_{\overline\Sigma}(\Gamma):(\gamma_{0}\SL\phi)|\Sigma\in H^{{\frac32}}(\Sigma)\}\,.
$$ 
By Theorem \ref{Theorem_Krein} (see \cite[Subsection 6.1]{MaPoSi}), 
\be\label{di-l}
\Delta_{D,\Sigma}u=\Delta u-[\hat\gamma_{1}]u\,\delta_{\overline\Sigma}\,,
\ee
\begin{align*}
&\dom(\Delta_{D,\Sigma})\\
=&\{u\in H^{1}(\RE^{n})\cap H^{0}_{\Delta}(\RE^n\backslash\Gamma): 
[\hat\gamma_{1}]u\in \dom(\Theta_{D,\Sigma}),\, 
(\gamma^{\-}_{0}u)|\Sigma=(\gamma^{\+}_{0}u)|\Sigma=0
\}
\end{align*}
is self-adjoint and
\begin{align*}
(-\Delta_{D,\Sigma}+z)^{-1}=&
(-\Delta+z)^{-1}-\SL_{z}\Pi_{\Sigma}'(R_{\Sigma}\gamma_{0}\SL_{z}\Pi_{\Sigma}')^{-1}R_{\Sigma}\gamma_{0}(-\Delta+z)^{-1}\,.
\end{align*}
Denoting by $\langle\cdot,\cdot\rangle_{-1,1}$ the $H^{-1}(\RE^{n})$-$H^{1}(\RE^{n})$ duality, for any $u\in\dom(\Delta_{D,\Sigma})\subset H^{1}(\RE^{n})$ one has, by \eqref{di-l} and $\langle\delta_{\Sigma}, u\rangle_{{-1},1}=0$ whenever $\supp(\gamma_{0}u)\subseteq \Sigma^{c}$, 
\begin{align*}
&\langle-\Delta_{D,\Sigma}u,u\rangle_{L^{2}(\RE^{n})}=
\langle-\Delta u,u\rangle_{{-1}, {1}}=\|\nabla u\|^{2}_{L^{2}(\RE^{n})}
\end{align*}
and so $\Delta_{D,\Sigma}\le 0$.\par
By Theorems \ref{SM} and \ref{negative1}, one gets, for any $k>0$,
\begin{align*}
s_{k}(\hat\xi,\hat \xi')
=&-\frac{i}{4\pi}\,\left(\frac{k}{2\pi}\right)^{n-2}\langle (R_{\Sigma}\gamma_{0}\SL^{+}_{-k^{2}}\Pi_{\Sigma}')^{-1}R_{\Sigma}\gamma_{0} u^{\circ}_{k\hat\xi'},R_{\Sigma}\gamma_{0} u^{\circ}_{k\hat\xi}\rangle\,.
\end{align*}

\end{subsection}  
\begin{subsection}{Neumann boundary conditions.}\label{neu-loc}
We denote by $\Delta_{N,\Sigma}$ the self-adjoint extension corresponding to the orthogonal projector defined by $\Pi(\phi\oplus\varphi):=0\oplus(\Pi_{\Sigma}\varphi)\equiv0\oplus(\varphi|\Sigma)$ and to the self-adjoint operator $\Theta(\phi\oplus\varphi):=0\oplus(-\Theta_{N,\Sigma}\varphi)$,
$$
\Theta_{N,\Sigma}:\dom( \Theta_{N,\Sigma})\subseteq H^{-{\frac12}}_{\overline\Sigma}(\Gamma)\to H^{{\frac12}}(\Sigma)\,,\quad \Theta_{N,\Sigma}\varphi=(\gamma_{1}\DL\varphi)|\Sigma\,,
$$
$$
\dom(\Theta_{N,\Sigma}):=\{\varphi\in H^{{\frac12}}_{\overline\Sigma}(\Gamma):(\gamma_{1}\DL\varphi)|\Sigma\in H^{{\frac12}}(\Sigma)\}\,.
$$ 
By Theorem \ref{Theorem_Krein} (see \cite[Subsection 6.2]{MaPoSi}), 
\be\label{ne-l}
\Delta_{N,\Sigma}u=\Delta u-[\gamma_{0}]u\,\nu\!\cdot\!\nabla\delta_{\overline\Sigma}\,,
\ee
\begin{align*}
\dom(\Delta_{N,\Sigma})
=&\{u\in H^{1}(\RE^{n}\backslash\overline\Sigma)\cap H^{0}_{\Delta}(\RE^n\backslash\Gamma): [\gamma_{0}]u\in \dom(\Theta_{N,\Sigma}),\\ &\quad [\hat\gamma_{1}]u=0,\, 
(\hat\gamma_{1}^{\-}u)|\Sigma=(\hat\gamma_{1}^{\+}u)|\Sigma=0\}
\end{align*}
is self-adjoint and
\begin{align*}
&(-\Delta_{N,\Sigma}+z)^{-1}\\=&
(-\Delta+z)^{-1}-\DL_{z}\Pi_{\Sigma}'(R_{\Sigma}\hat\gamma_{1}\DL_{z}\Pi_{\Sigma}')^{-1}R_{\Sigma}\gamma_{1}(-\Delta+z)^{-1}\,.
\end{align*}
By    Green's formula \eqref{hG}, for any $u\in\dom(\Delta_{N,\Sigma})\subset H^{1}(\RE^{n}\backslash\Gamma)\cap H^{0}_{\Delta}(\RE^{n}\backslash\Gamma)$ one has 
\begin{align*}
\langle-\Delta_{N,\Sigma}u,u\rangle_{L^{2}(\RE^{n})}=&
\|\nabla u\|^{2}_{L^{2}(\Omega_{\-})}+\|\nabla u\|^{2}_{L^{2}(\Omega_{\+})}+\langle\hat\gamma_{1}u,[\gamma_{0}]u\rangle_{L^{2}(\Gamma)}\\
=&\|\nabla u\|^{2}_{L^{2}(\Omega_{\-})}+\|\nabla u\|^{2}_{L^{2}(\Omega_{\+})}
\end{align*}
and so $\Delta_{N,\Sigma}\le 0$.\par
Then, by Theorems \ref{SM} and \ref{negative1}, one gets, for any $k>0$,
\begin{align*}
s_{k}(\hat\xi,\hat \xi')
=&-\frac{i}{4\pi}\,\left(\frac{k}{2\pi}\right)^{n-2}\langle (R_{\Sigma}\hat\gamma_{1}\DL^{+}_{-k^{2}}\Pi_{\Sigma}')^{-1}R_{\Sigma}\gamma_{1} u^{\circ}_{k\hat\xi'},R_{\Sigma}\gamma_{1} u^{\circ}_{k\hat\xi}\rangle\,.
\end{align*}
\end{subsection}  
\begin{subsection}{Robin boundary conditions.}\label{robin-loc}  
We denote by $\Delta_{R,\Sigma}$ the self-adjoint extension corresponding to the orthogonal projector  defined by $\Pi(\phi\oplus\varphi):=\Pi^{\oplus}_{\Sigma}(\phi\oplus\varphi)=(\Pi_{\Sigma}\phi)\oplus(\Pi_{\Sigma}\varphi)\equiv R_{\Sigma}^{\oplus}(\phi\oplus\varphi):=(\phi|\Sigma)\oplus(\varphi|\Sigma)$ and to the self-adjoint operator $\Theta:=-\Theta_{R,\Sigma}$,
\begin{align*}
\Theta_{R,\Sigma}:
\dom(\Theta_{R,\Sigma})\subseteq 
H^{-\frac32}_{\overline\Sigma}(\Gamma)\oplus H^{-\frac12}_{\overline\Sigma}(\Gamma)\to H^{\frac32}(\Sigma)\oplus H^{\frac12}(\Sigma)\,,
\end{align*}
\begin{align*}
\Theta_{R,\Sigma}(\phi,\varphi)
=&
(((1/[b]+\gamma_{0}\SL)\phi+(\langle b\rangle/[b]+\hat\gamma_{0}\DL)\varphi)|\Sigma)\\&\oplus(((\langle b\rangle/[b]+\hat\gamma_{1}\SL)\phi+(b_{\+}b_{\-}/[b]+\hat\gamma_{1}\DL)\varphi)|\Sigma)
\,,
\end{align*}
\begin{align*}
&\dom(\Theta_{R,\Sigma}):=\big\{(\phi,\varphi)\in L^{2}_{\overline\Sigma }(\Gamma)\times H^{\frac12}_{\overline\Sigma }(\Gamma):\\
&((1/[b]+\gamma_{0}\SL)\phi+(\langle b\rangle/[b]+\hat\gamma_{0}\DL)\varphi)|\Sigma\in H^{\frac32}(\Sigma)\,,\\
&((\langle b\rangle/[b]+\hat\gamma_{1}\SL)\phi
+(b_{\+}b_{\-}/[b]+\hat\gamma_{1}\DL)\varphi)|\Sigma\in H^{\frac12}(\Sigma)\big\}\,.
\end{align*}
Here $b_{\-}$ and $b_{\+}$ satisfy the same hypotheses as in subsection \ref{robin} and $b_{\-} > b_{\+}$. By Theorem \ref{Theorem_Krein}  (see \cite[Subsection 6.3]{MaPoSi}),
\begin{align*}
&\Delta_{R,\Sigma}u=\Delta u-\frac4{[b]}\,\big((\langle b\rangle\,\gamma_{1}u-b_\+b_{\-}\gamma_{0}u)\,\delta_{\overline\Sigma}+(\gamma_{1}u-\langle b\rangle\,\gamma_{0}u)\,\nu\!\cdot\!\nabla\delta_{\overline\Sigma}\big)\,,
\end{align*}
\begin{align*}
\dom(\Delta_{R,\Sigma})
=&\{u\in H^{1}(\RE^{n}\backslash\overline\Sigma)\cap
H^{0}_{\Delta}(\RE^{n}\backslash\Gamma)
: [\hat\gamma]u\in \dom(\Theta_{R,\Sigma})\,,\\
&\ (\gamma_{1}^\- u-b_{\-}\gamma_{0}^{\-}u)|\Sigma=(\gamma_{1}^\+ u-b_{\+}\gamma_{0}^{\+}u)|\Sigma=0\}
\end{align*}
is self-adjoint and
\begin{align*}
&(-\Delta_{R,\Sigma}+z)^{-1}-(-\Delta+z)^{-1}\\
=&
-G_{z}(\Pi^{\oplus}_{\Sigma})'\left(R ^{\oplus}_{\Sigma}\left[\,\begin{matrix}1/{[b]}+\gamma_{0}\SL_{z}&\langle b\rangle/[b]+\gamma_{0}\DL_{z}\\
\langle b\rangle/[b]+\gamma_{1}\SL_{z}&b_{+}b_{-}/[b]+\gamma_{1}\DL_{z}
\end{matrix}\,\right](\Pi^{\oplus}_{\Sigma})'\right)^{ -1}\\
&\times R^{\oplus}_{\Sigma}
\gamma(-\Delta+z)^{-1}\,,
\end{align*}
where $(\Pi^{\oplus}_{\Sigma})'$ is the orthogonal projection onto 
$H_{\overline\Sigma}^{-\frac32}(\Gamma)\oplus 
H_{\overline\Sigma}^{-\frac12}(\Gamma)$
and $G_{z}$ is defined in \eqref{Gz}. 
\begin{remark} By \cite[Remark 6.15]{MaPoSi}, $\Delta_{R,\Sigma}$ depends only on $\Sigma$, $b_{\-}|\Sigma$ and $b_{\+}|\Sigma$. Thus, by considering $\t\Omega\subset\RE^{n}$ such that $\t\Omega_{\-}=\t\Omega\subset\Omega_{\+}$, $\Omega_{\-}\subset\t\Omega_{\+}=\RE^{n}\backslash\overline{\t\Omega}$ and $\Sigma\subset \t\Gamma=\partial\t\Omega$, it is possible to convert the assumption $b_{\-}>b_{\+}$ into $b_{\+}>b_{\-}$. 
\end{remark}
By    Green's formula \eqref{hG} and by Ehrling's lemma, for any $u\in\dom(\Delta_{R,\Sigma})\subset H^{1}(\RE^{n}\backslash\Gamma)\cap H^{0}_{\Delta}(\RE^{n}\backslash\Gamma)$ one has (here the Sobolev index $s$ belongs to $(\frac12,1)$)
\begin{align*}
&\langle-\Delta_{R,\Sigma}u,u\rangle_{L^{2}(\RE^{n})}\\=&
\|\nabla u\|^{2}_{L^{2}(\Omega_{\-})}+\|\nabla u\|^{2}_{L^{2}(\Omega_{\+})}\\&-\langle b_{\-}\gamma^{\-}_{0}u_{\-},\gamma^{\-}_{0}u_{\-}\rangle_{L^{2}(\Sigma)}+
\langle b_{\+}\gamma^{\+}_{0}u_{\+},\gamma^{\+}_{0}u_{\+}\rangle_{L^{2}(\Sigma)}
\\
\ge&\|\nabla u\|^{2}_{L^{2}(\Omega_{\-})}+\|\nabla u\|^{2}_{L^{2}(\Omega_{\+})}\\&-\big(\|b_{\-}\|_{L^{\infty}(\Gamma)}+\|
b_{\+}\|_{L^{\infty}(\Gamma)}\big)\big(\|\gamma^{\-}_{0}u_{\-}\|^{2}_{L^{2}(\Gamma)}+\|\gamma^{\+}_{0}u_{\+}\|^{2}_{L^{2}(\Gamma)}\big)\\
\end{align*}
\begin{align*}
\ge&\|\nabla u\|^{2}_{L^{2}(\Omega_{\-})}+\|\nabla u\|^{2}_{L^{2}(\Omega_{\+})}\\&-c\,
\big(\|b_{\-}\|_{L^{\infty}(\Gamma)}+
\|b_{\+}\|_{L^{\infty}(\Gamma)}\big)\big(\|u_{\-}\|^{2}_{H^{s}(\Omega_{\-})}+\|u_{\+}\|^{2}_{H^{s}(\Omega_{\+})}\big)\\
\ge&\|\nabla u\|^{2}_{L^{2}(\Omega_{\-})}+\|\nabla u\|^{2}_{L^{2}(\Omega_{\+})}\\&-c\,\big(\|b_{\-}\|_{L^{\infty}(\Gamma)}+\|
b_{\+}\|_{L^{\infty}(\Gamma)}\big)\big(\epsilon\big(\|u_{\-}\|^{2}_{H^{1}(\Omega_{\-})}+\|u_{\+}\|^{2}_{H^{1}(\Omega_{\-})}\big)+c_{\epsilon}\|u\|^{2}_{L^{2}(\RE^{n})}\big)\\
\ge&-\kappa_{\epsilon}\|u\|^{2}_{L^{2}(\RE^{n})}
\end{align*}
and so and so $\Delta_{\beta,\delta',\Sigma}\le \kappa_{\epsilon}$. 
Then, by Theorems \ref{SM} and \ref{negative1}, one gets, for any $k>0$,
\begin{align*}
&s_{k}(\hat\xi,\hat \xi')\\
=&-\frac{i}{4\pi}\left(\frac{k}{2\pi}\right)^{\!\!n-2}\!\!\Bigg\langle\!\! \left(\!R^{\oplus}_{\Sigma}
\!\left[\begin{matrix}1/{[b]}+\gamma_{0}\SL^{+}_{-k^{2}}&\langle b\rangle/[b]+\gamma_{ 0}\DL^{+}_{-k^{2}}\\
\langle b\rangle/[b]+\gamma_{1}\SL^{+}_{-k^{2}}&b_{+}b_{-}/[b]+\gamma_{1}\DL^{+}_{-k^{2}}
\end{matrix}\!\right]\!(\Pi^{\oplus}_{\Sigma})'\right)^{\!\! -1}\\
&\times R^{\oplus}_{\Sigma}
\gamma u^{\circ}_{k\xi'},R^{\oplus}_{\Sigma}\gamma u^{\circ}_{k\xi}\Bigg\rangle.
\end{align*}
\end{subsection}  
\begin{subsection}{$\delta$-interactions.}\label{delta-loc} We denote by $\Delta_{\alpha,\delta,\Sigma}$ the self-adjoint extension corresponding to the orthogonal projector defined by $\Pi(\phi\oplus\varphi):=(\Pi_{\Sigma}\phi)\oplus 0\equiv(\phi|\Sigma)\oplus 0$ and to the self-adjoint operator 
$\Theta(\phi\oplus\varphi):=(-\Theta_{\alpha,D,\Sigma}\phi)\oplus 0$,
$$\Theta_{\alpha,D,\Sigma}:\dom(\Theta_{\alpha,D,\Sigma})
\subseteq H^{-\frac32}_{\overline\Sigma}(\Gamma)\to H^{\frac32}(\Sigma)\,,\quad\Theta_{\alpha,D,\Sigma}\phi:=((1/\alpha+\gamma_{0}\SL)\phi)|\Sigma\,,
$$
$$
\dom(\Theta_{\alpha,D,\Sigma}):=\{\phi\in L^{2}_{\overline\Sigma }(\Gamma): ((1/\alpha+\gamma_{0}\SL)\phi)|\Sigma\in H^{\frac32}(\Sigma)\}\,.
$$
Here $\alpha$ satisfies the same hypothesis as in subsection \ref{delta} and we further require that it has constant sign
on (each component of) $\Gamma$. By Theorem \ref{Theorem_Krein} (see \cite[Subsection 6.4]{MaPoSi}; notice a misprint in \cite[Corollary 6.21]{MaPoSi}: the condition $[\gamma_{0}]u=0$, implying $u\in H^{1}(\RE^{n})$, is missing in $\dom(\Delta_{\alpha,\delta,\Sigma})$), 
\begin{align*}
&\Delta_{\alpha,\delta,\Sigma}\,u=\Delta u-\alpha\gamma_{0}u\,\delta_{\overline\Sigma}
\end{align*}
\begin{align*}
&\dom(\Delta_{\alpha,\delta,\Sigma})\\
=&\{u\in H^{1}(\RE^{n})\cap H^{2-}(\RE^{n}\backslash\overline\Sigma)\cap H^{0}_{\Delta}(\RE^{n}\backslash\Gamma): [\gamma_{1}]u\in \dom(\Theta_{\alpha,D,\Sigma})\,,\\&\  (\alpha\gamma_{0}u-[\gamma_{1}]u)|\Sigma=0\}
\end{align*}
is self-adjoint and \begin{align*}
&(-\Delta_{\alpha,\delta,\Sigma}+z)^{-1}u\\
=&
(-\Delta+z)^{-1}-\SL_{z}\Pi_{\Sigma}'(R_{\Sigma}(1+\alpha\gamma_{0}\SL_{z})\Pi_{\Sigma}')^{-1}R_{\Sigma}\alpha\gamma_{0}(-\Delta+z)^{-1}\,.
\end{align*}
For any $u\in\dom(\Delta_{\alpha,\delta,\Sigma})$, by Ehrling's lemma, one has (here the Sobolev index $s$ belongs to $(\frac12,1)$)
\begin{align*}
\langle-\Delta_{\alpha,\delta,\Sigma}u,u\rangle_{L^{2}(\RE^{n})}=&
\langle-\Delta u,u\rangle_{{-1}, {1}}+\langle\alpha\gamma_{0}u, \gamma_{0}u\rangle_{L^{2}(\Sigma)}\\
=&\|\nabla u\|^{2}_{L^{2}(\RE^{n})}+\langle\alpha\gamma_{0}u, \gamma_{0}u\rangle_{L^{2}(\Sigma)}\\
\ge& \|\nabla u\|^{2}_{L^{2}(\RE^{n})}-\|\alpha\|_{L^\infty(\Gamma)}\|\gamma_{0}u\|^{2}_{L^{2}(\Gamma)}\\
\ge&\|\nabla u\|^{2}_{L^{2}(\RE^{n})}-c\,\|\alpha\|_{L^\infty(\Gamma)}\|u\|^{2}_{H^{s}(\Omega)}
\\
\ge&\|\nabla u\|^{2}_{L^{2}(\RE^{n})}-c\,\|\alpha\|_{L^\infty(\Gamma)}\big(\epsilon\,\|u\|^{2}_{H^{1}(\Omega)}+c_{\epsilon}\|u\|^{2}_{L^{2}(\Omega)}\big)\\
\ge&-\kappa_{\epsilon}\|u\|^{2}_{L^{2}(\Omega)}
\end{align*}
and so $\Delta_{\alpha,\delta,\Sigma}\le \kappa_{\epsilon}$. 
By Theorems \ref{SM} and \ref{negative1}, one gets, for any $k>0$,
\begin{align*}
s_{k}(\hat\xi,\hat \xi')
=&-\frac{i}{4\pi}\,\left(\frac{k}{2\pi}\right)^{n-2}\langle (R_{\Sigma}(1+\alpha\gamma_{0}\SL_{z})\Pi_{\Sigma}')^{-1}R_{\Sigma}\alpha\gamma_{0} u^{\circ}_{k\hat\xi'},R_{\Sigma}\gamma_{0} u^{\circ}_{k\hat\xi}\rangle\,.
\end{align*}
\end{subsection}  
\begin{subsection} {$\delta'$-interaction.}\label{delta'-loc} We denote by $\Delta_{\beta,\delta',\Sigma}$ the self-adjoint extension corresponding to the orthogonal projector defined by $\Pi(\phi\oplus\varphi):=0\oplus(\Pi_{\Sigma}\varphi)\equiv0\oplus(\varphi|\Sigma)$ and to the self-adjoint operator $\Theta(\phi\oplus\varphi):=0\oplus(-\Theta_{\beta,N,\Sigma})\varphi$,
$$\Theta_{\beta,N,\Sigma}:
\dom(\Theta_{\beta,N,\Sigma})\subseteq 
H^{-\frac12}_{\overline\Sigma}(\Gamma)\to  H^{\frac12}(\Sigma)\,,\quad
$$
$$\Theta_{\beta,N,\Sigma}\phi:=
((-1/\beta+\hat\gamma_{1}\DL)\phi)|\Sigma\,,
$$
$$
\dom(\Theta_{\beta,N,\Sigma}):=\{\varphi\in H^{\frac12}_{\overline\Sigma }(\Gamma):
((-1/\beta+\hat\gamma_{1}\DL)\varphi)|\Sigma \in H^{\frac12}(\Sigma)\}\,.
$$
Here $\beta$ satisfies the same hypothesis as in subsection \ref{deltaprimo}.
By Theorem \ref{Theorem_Krein} (see \cite[Subsection 6.5]{MaPoSi}; notice a misprint in \cite[Corollary 6.26]{MaPoSi}: the condition $[\hat\gamma_{1}]u=0$ is missing in $\dom(\Delta_{\beta,\delta',\Sigma})$),
\begin{align*}
\Delta_{\beta,\delta',\Sigma}u=\Delta u-\beta\gamma_{1}u\,\nu\!\cdot\!\nabla
\delta_{\overline\Sigma}\,,
\end{align*}
\begin{align*}
&\dom(\Delta_{\beta,\delta',\Sigma})\\
=&\{u\in H^{1}(\RE^{n}\backslash\overline\Sigma)\cap H^{0}_{\Delta}(\RE^{n}\backslash\Gamma): [\gamma_{0}]u\in \dom(\Theta_{\beta,N,\Sigma})\,,\ [\hat\gamma_{1}]u=0\,,\\&\  (\beta\hat\gamma_{1}u-[\gamma_{0}]u)|\Sigma=0\}
\end{align*}
is self-adjoint and its resolvent is given by 
\begin{align*}
&(-\Delta_{\beta,\delta',\Sigma}+z)^{-1}u\\
=&
(-\Delta+z)^{-1}+\DL_{z}\Pi_{\Sigma}'(R_{\Sigma}(1-\beta\hat\gamma_{1}\DL_{z})\Pi_{\Sigma}')^{-1}R_{\Sigma}\beta\gamma_{1}(-\Delta+z)^{-1}\,.
\end{align*}
By    Green's formula \eqref{hG} and by Ehrling's lemma, for any $u\in\dom(\Delta_{\beta,\delta',\Sigma})\subset H^{1}(\RE^{n}\backslash\Gamma)\cap H^{0}_{\Delta}(\RE^{n}\backslash\Gamma)$ one has (here the Sobolev index $s$ belongs to $(\frac12,1)$)
\begin{align*}
&\langle-\Delta_{\beta,\delta',\Sigma}u,u\rangle_{L^{2}(\RE^{n})}\\=&
\|\nabla u\|^{2}_{L^{2}(\Omega_{\-})}+\|\nabla u\|^{2}_{L^{2}(\Omega_{\+})}+\langle(1/\beta)[\gamma_{0}]u,[\gamma_{0}]u\rangle_{L^{2}(\Sigma)}\\
\ge&\|\nabla u\|^{2}_{L^{2}(\Omega_{\-})}+\|\nabla u\|^{2}_{L^{2}(\Omega_{\+})}-2\,\|1/\beta\|_{L^\infty(\Gamma)}\big(\|\gamma^{\-}_{0}u_{\-}\|^{2}_{L^{2}(\Gamma)}+\|\gamma^{\+}_{0}u_{\+}\|^{2}_{L^{2}(\Gamma)}\big)\\
\ge&\|\nabla u\|^{2}_{L^{2}(\Omega_{\-})}+\|\nabla u\|^{2}_{L^{2}(\Omega_{\+})}-c\,\|1/\beta\|_{L^\infty(\Gamma)}\big(\|u_{\-}\|^{2}_{H^{s}(\Omega_{\-})}+\|u_{\+}\|^{2}_{H^{s}(\Omega_{\+})}\big)\\
\ge&\|\nabla u\|^{2}_{L^{2}(\Omega_{\-})}+\|\nabla u\|^{2}_{L^{2}(\Omega_{\+})}-c\,\|1/\beta\|_{L^\infty(\Gamma)}\big(\epsilon\big(\|u_{\-}\|^{2}_{H^{1}(\Omega_{\-})}+\|u_{\+}\|^{2}_{H^{1}(\Omega_{\-})}\big)\\
&+c_{\epsilon}\|u\|^{2}_{L^{2}(\RE^{n})}\big)
\ge-\kappa_{\epsilon}\|u\|^{2}_{L^{2}(\RE^{n})}
\end{align*}
and so and so $\Delta_{\beta,\delta',\Sigma}\le \kappa_{\epsilon}$. 
Then, by Theorems \ref{SM} and \ref{negative1}, one gets, for any $k>0$,
\begin{align*}
s_{k}(\hat\xi,\hat \xi')
=&-\frac{i}{4\pi}\,\left(\frac{k}{2\pi}\right)^{n-2}\langle (R_{\Sigma}(1-\beta\hat\gamma_{1}\DL^{+}_{-k^{2}})\Pi_{\Sigma}')^{-1}R_{\Sigma}\beta\gamma_{1} u^{\circ}_{k\hat\xi'},R_{\Sigma}\gamma_{1} u^{\circ}_{k\hat\xi}\rangle\,.
\end{align*}
\end{subsection}

\end{section}

\end{document}